\documentclass[10pt,twocolumn,twoside]{IEEEtran}
\usepackage{graphicx}
\usepackage[cmex10]{amsmath}

\usepackage{amssymb}
\usepackage{array}

\usepackage{amsfonts}

\usepackage{color}

\usepackage{setspace}

\usepackage{amsmath,amsfonts}
\usepackage{amssymb,bbm,bm}
\usepackage{graphicx}
\newcommand{\EXP}[1]{\mathsf{E}\left[ #1 \right]}
\newtheorem{lemma}{\textbf{Lemma}}

\newtheorem{theorem}{\textbf{Theorem}}
\newtheorem{remark}{\textbf{Remark}}

\newcommand{\Reals}{{\mathbb{R}}}

\begin{document}

\title{Revenue Maximization in Service Systems with Heterogeneous Customers}

\author{
  \begin{tabular}{ccc}
    Tejas Bodas   & D.~Manjunath \\
    IIT Bombay, INDIA & IIT Bombay, INDIA 
  \end{tabular}
}

\maketitle{}

\begin{abstract}
In this paper, we consider revenue maximization problem 
for a two server system in the presence of heterogeneous
customers. We assume that the customers differ in their cost
for unit delay and this is modeled as a continuous random
variable with a distribution $F.$ We also assume that each
server charges an admission price to each customer that
decide to join its queue. We first consider the monopoly
problem where both the servers
belong to a single operator. The heterogeneity of the customer
makes the analysis of the problem difficult. The difficulty 
lies in the inability to characterize the equilibrium queue 
arrival rates as a function of the admission prices. We 
provide an equivalent formulation with the queue arrival rates
as the optimization variable simplifying the 
analysis for revenue rate maximization for the monopoly. 
We then consider the duopoly problem where each server 
competes with the other server to maximize its revenue rate.
For the duopoly problem, the interest is to obtain the set of 
admission prices satisfying the Nash equilibrium conditions.  
While the problem is in general difficult to analyze, we 
consider the special case when the two servers are identical.
For such a duopoly system, we obtain the necessary condition 
for existence of symmetric Nash equilibrium of the admission prices.
The knowledge of the distribution $F$ characterizing the 
heterogeneity of the customers is necessary to solve the monopoly
and the duopoly problem. However, for most practical 
scenarios, the functional form of $F$ may not be known
to the system operator and in such cases, the
revenue maximizing prices cannot be determined. In the last 
part of the paper, we provide a simple method to estimate the 
distribution $F$ by suitably varying the admission prices.  
We illustrate the method with some numerical examples.
 \end{abstract}

 


\section{Introduction}
\label{sec:intro}

In many service systems, the quality of service received
is characterized by the queueing delay that is experienced
by the customers in the system.  Examples of such service 
systems that can be modeled as queueing systems include 
road and transport systems, health-care systems, computer 
systems, call centers and communications systems. The 
customers that receive service in such systems are usually sensitive
to the delay experienced in these system. Further, such 
customers have non-identical preferences to the delay 
experienced. It is often beneficial for the service 
system to account for these heterogeneous preferences in 
any optimization concerning the use of system resources.
Many service systems have emerged that exploit the 
heterogeneous nature of customers and use it to their advantage.
For example, airlines offer priority boarding queues for payment
of an additional fee. In this paper, we consider the 
problem of exploiting the heterogeneous nature of customers
for revenue maximization in parallel server systems.
We model heterogeneity of customers by assuming 
that different customers have different cost for a unit delay.

We consider service systems that consist of two parallel,
possibly heterogeneous servers where each server has 
an associated queue for the customers to wait. The scheduling
discipline at each server is work conserving and does 
not discriminate between customers on the basis of their
preference for delay. The servers charge an admission price 
to every customer joining its queue. We assume that 
the queues are not observable and only the expected 
delay as a function of the arrival rate is available.
We also assume that the expected delay at any server is monotone 
increasing in the arrival rate of customers to that server.
The customers that use the system are strategic and make an
individually optimal queue-join decision. We assume that 
customers differ in their cost for unit delay which is 
characterized by a random variable  with a continuous 
distribution denoted by $F.$ 
For a customer, the cost at a server is the sum
of the admission price and the delay cost at the server.
We assume that customers cannot balk from the system
without obtaining service and such traffic is commonly
seen in cloud-computing, purchase of essential services etc.

In this paper, we consider the problem of revenue maximization 
in such a service system by suitably choosing the admission 
prices at two parallel servers. Depending on  the objective
of each of these servers, we consider two natural scenarios.
In a monopoly, we assume that the two servers belong to the 
same operator. The objective here is
to maximize the total revenue rate, i.e., the sum of the revenue rate
from the two servers. In the second scenario, we assume that each 
server belongs to separate operators and each server has the objective of
maximizing its individual revenue rate. This is an example of a
duopoly where the service systems compete with one
another to maximize their individual revenue rate.

Now consider the scenario of a monopoly market discussed above
where the service system has two parallel servers.
In the absence
of balking, it is not difficult to see that a revenue maximizing strategy for 
the monopoly is to keep both the admission prices at infinity. This is because
as customers cannot balk, they are required to choose one of the server for 
service. Therefore one has to consider a more meaningful
model for the monopoly market. Towards this, we assume that the
admission price at one of the server, say Server~2  is 
fixed a-priori. This dissuades the service provider from fixing the admission
price at Server~1 to unreasonably high values. Our interest for this 
model is to characterize the revenue
maximizing admission price at Server~1 for different examples of the delay 
functions at the queue and when customers differ in their delay cost.

Classical monopoly models have been well studied for the case of  
single server queues. One of the first work to analyze such
a model is Naor \cite{Naor69}. This model considers
a single server queueing system where homogeneous customers 
obtain a reward after service completion. The queue is 
observable to arriving customers who choose to either join the 
queue or balk. For such a system, the revenue maximizing 
admission price was first obtained in \cite{Naor69}. Subsequently,
there have been several works analyzing the revenue 
maximization problem for various models such as a
multiserver queue \cite{Knudsen72},
$GI/M/1$ queue \cite{Yechiali71}, customers
with heterogeneous service valuations \cite{Larsen98}
and queue length dependent prices  
\cite{Chen01}. While the above models assume that the queue 
lengths are observable, Edelson and Hilderbrand \cite{Edelson75} 
were the first to consider the revenue maximization problem for 
the case  when queues are not observable. 
See \cite{Mendelson85,Mendelson90,Bradford96,Masuda99,Chen04}
for some other single server revenue maximization models.

The key difference of our model with that of the literature 
discussed above is as follows. Firstly, in our model, customers are 
inelastic in their demand and hence balking is not allowed.
Secondly, the customers
have to obtain service at either of the two servers and
the admission price at one of the server is fixed. 
Finally, the customers have heterogeneous preference for the
delay experienced in the queue. This feature 
makes our model meaningful but also difficult to analyze.
For such parallel server models, the structural 
properties for the equilibrium routing have been obtained
recently \cite{Bodas11b,Bodas14}. We use the structural property of 
the equilibrium routing to solve the the revenue maximization 
problem for the monopoly. 

For the duopoly problem with two competing and identical
servers, we assume that the objective for each server is to set an 
admission price that maximizes its revenue rate.  
We are interested in studying the existence of Nash
Equilibrium prices that would be set by the two servers.
The earliest work analyzing the duopoly model
with heterogeneous customers was by Luski \cite{Luski76} and Levhari and
Luski \cite{Levhari78}. Both the models assume that the customers are allowed to balk.
Luski \cite{Luski76}
is interested in knowing whether the revenue maximizing prices set by the two service 
systems can be equal. It is observed that when the parameters of the model
are such that the customers have no incentive to balk, the revenue maximizing 
prices set by two identical servers is equal. This is however not the case when
some of the customers prefer to balk. In this case, the equilibrium revenue
maximizing prices are not equal. Levhari and Luski \cite{Levhari78} provide
a numerical analysis for the problem introduced in Luski \cite{Luski76}. 
Armony and Haviv \cite{Armony03} analyze this problem for the 
case when the customers are from a finite number of classes and each 
class has a distinct cost for unit delay. A numerical 
analysis of the Nash equilibrium admission prices between the two competing
servers is provided.
Chen and Wan \cite{Chen03} consider the revenue maximization in a duopoly 
with a single customer class. The service system is modeled by  $M/M/1$ 
queues and the customers are allowed to balk from the system. These 
assumptions on the system model allows them to obtain the sufficient 
conditions for the existence of Nash equilibrium. Similar conditions were 
found in Dube and Jain \cite{Dube08} who consider 
an $N$-player oligopoly with multiclass customers. The customer classes
differ only in their arrival rates and have the same delay cost per
unit time. A differentiated service model
is considered by Dube and Jain \cite{Dube10} where each player now operates two types 
of services and each service is used by a dedicated class of customers.
Again, the key result in \cite{Dube10} is to obtain the sufficient condition
for the Nash equilibrium prices. Mandjes and Timmers \cite{Mandjes07}
consider a duopoly model with two customer classes differing in their 
delay cost. The model assumes a finite number of customers 
and the utility of a queue is a decreasing function of the number of 
customers using this server. Given the prices at the servers, they provide
an algorithm that determines the equilibrium number of customers of each
class that is to be allocated to the two servers. While the existence 
and uniqueness of such a customer equilibrium is provided, the existence
of Nash equilibrium prices is only conjectured. In \cite{Allon08,Allon07}
the demand rate at different servers is modeled using specific functions
(known as demand models in such literature) instead of being calculated
from the (Wardrop) equilibrium conditions \cite{Wardrop52}. This assumptions make the 
analysis relatively simpler. Ayesta et. al. \cite{Ayesta11} consider
the oligopoly pricing game for a
single customer class and obtain the necessary and sufficient conditions on the Nash 
equilibrium prices when the queues have identical delay functions. 
A best-response algorithm is then provided to numerically obtain these
Nash equilibrium prices.

Most of the monopoly and duopoly models described above, 
make simplifying assumptions on the customer classes to characterize the underlying
Wardrop equilibrium \cite{Wardrop52}. Additional simplification of the  
analysis is obtained by considering convex and increasing delay
functions at the queues. We do not make any of these assumptions in 
this paper. We utilize the structure of the Wardrop equilibrium 
that was characterized in \cite{Bodas11b,Bodas14} to analyze the two 
problems. This structure on the equilibrium allows us to provide an 
equivalent revenue maximization formulation for both the monopoly
and the duopoly that is simpler to analyze. For the duopoly problem we 
provide  sufficient conditions on the symmetric 
Nash equilibrium prices when the competing servers are identical. 

For most practical scenarios, the distribution function $F(\cdot)$ 
characterizing the delay cost for a customer  may 
not be known to the service system. The revenue maximizing 
strategy on the other hand depend on the distribution $F(\cdot).$ Without
any knowledge of $F(\cdot),$  it is  not be possible to ascertain a revenue optimal 
admission price at the servers and in such cases, the service system is 
required estimate this distribution function. Towards the end of this paper,
we shall provide a simple method to estimate this distribution $F(\cdot)$
by varying the admission prices and observing the change in the equilibrium
traffic routing. The service system can then use this estimate to perform 
the necessary revenue maximization.

The rest of the paper is organized as follows. In the next section, we shall
formalize the notations and provide some preliminaries. We then formulate the revenue
maximization problems in Section \ref{sec:formulation}. 
In Section \ref{sec:monopoly}, we consider the monopoly problem for
revenue optimization followed by the duopoly problem in Section \ref{sec:duopoly}.
Finally in Section \ref{sec:estimate_F}, we illustrate a mechanism based on 
admission pricing to estimate the distribution function $F.$

\section{Preliminaries}
\label{sec:prelim_2}
We will first introduce the notations that will be used throughout 
this paper. In both the monopoly and
the duopoly model, we assume that the system has two servers.
Let $c_j$ denote the admission price at Server~$j$ where 
$ j = 1, 2.$ The customers arrive according to a homogeneous
Poisson process with rate $\lambda$ and have a service requirement
that is i.i.d with exponential distribution and unit mean.
Let $D_j(\gamma_j)$ denote the delay function
associated with queue~$j$ when the queue arrival rate is 
$\gamma_j,$ where $j = 1, 2.$ Note that $\gamma_1 + \gamma_2 = \lambda.$
We assume that $D_j$ is monotone
increasing and continuously differentiable in the interior of
its domain with a strictly positive derivative. Additionally
we assume that the cost function at the two server satisfies
the following two conditions (1) $D_1(0)<D_2(\lambda)<\infty$
and (2)  $D_2(0)<D_1(\lambda)<\infty.$ 

We associate with each arriving customer a continuous random variable $\bm{\beta}$ 
that quantifies a customer's sensitivity to delay or congestion.
We shall assume that the delay sensitivity ${\beta}$ for a customer 
is a realization of the random variable $\bm{\beta}.$ The customer arrivals
constitute a marked Poisson process of intensity $\lambda \times F$ on
$\Reals \times \Reals_+.$ Here $F$ is an absolutely continuous 
cumulative distribution function supported on the interval $[a,b]$ 
of positive reals.
We additionally assume that $F(\cdot)$
is strictly increasing  and hence $f(x) \neq 0$ 
for any $x \in [a,b]$ where $f(\cdot)$ is the corresponding 
density function.

We now recall the Wardrop equilibrium conditions \cite{Wardrop52,Bodas14}
that characterize the individually optimal choice of server made by the
arriving customers. A customer with delay cost $\beta$ entering the system
must choose a queue $j$ so as to minimize $c_j+\beta D_j(\gamma_j).$
Here $\gamma_j$ is determined through the strategies of all customers.
We assume that the quantities $\lambda_1, \lambda_2, D_j(\cdot), F(\cdot)$ and $c_j,$
for $j=1,2$ is part of common knowledge. We also assume
that the customers do not have access to current or past queue 
occupancies, or the history of arrival times. The strategy of a customer
is restricted to choosing a server according to a fixed probability 
distribution and such joint strategies are represented by a stochastic
kernel, denoted by $K^W.$ We interpret $K^W(\beta,i)$ as the probability
that a customer with delay sensitivity $\beta$ chooses queue $i$ at
equilibrium. For the two server system, the equilibrium kernel $K^W$ must
satisfy the following Wardrop equilibrium conditions.
\begin{equation}
\label{eq:wardrop_cond}
 K^W(\beta,i) \geq 0 \mbox{~implies~} c_i+\beta D_i(\gamma_i) \leq c_{3-i}+\beta D_{3-i}(\gamma_{3-i}).
\end{equation}
In words, this means that if customers with delay cost $\beta$ choose 
Server~$i$ at equilibrium, then the expected cost for this customer 
at Server~$i$ must be at most the expected cost at Server~$3-i$ for 
$i = 1,2.$ 
For a kernel $K^W,$ note that the arrival rate of customers to 
Server~$j$  is given by
$$
\gamma_j = \lambda \int_{\beta=a}^{b} K^W(\beta,j) dF(\beta).
$$
%

We now provide the following theorem that is a restatement of 
Corollary 4 in \cite{Bodas14}. This theorem characterizes
 the Wardrop equilibrium kernel for a system with two parallel servers. 
\begin{figure*}
  \begin{minipage}{6cm}
    \begin{center}
    \vspace{-3mm}
      \includegraphics[width=.45\columnwidth, height=6cm]{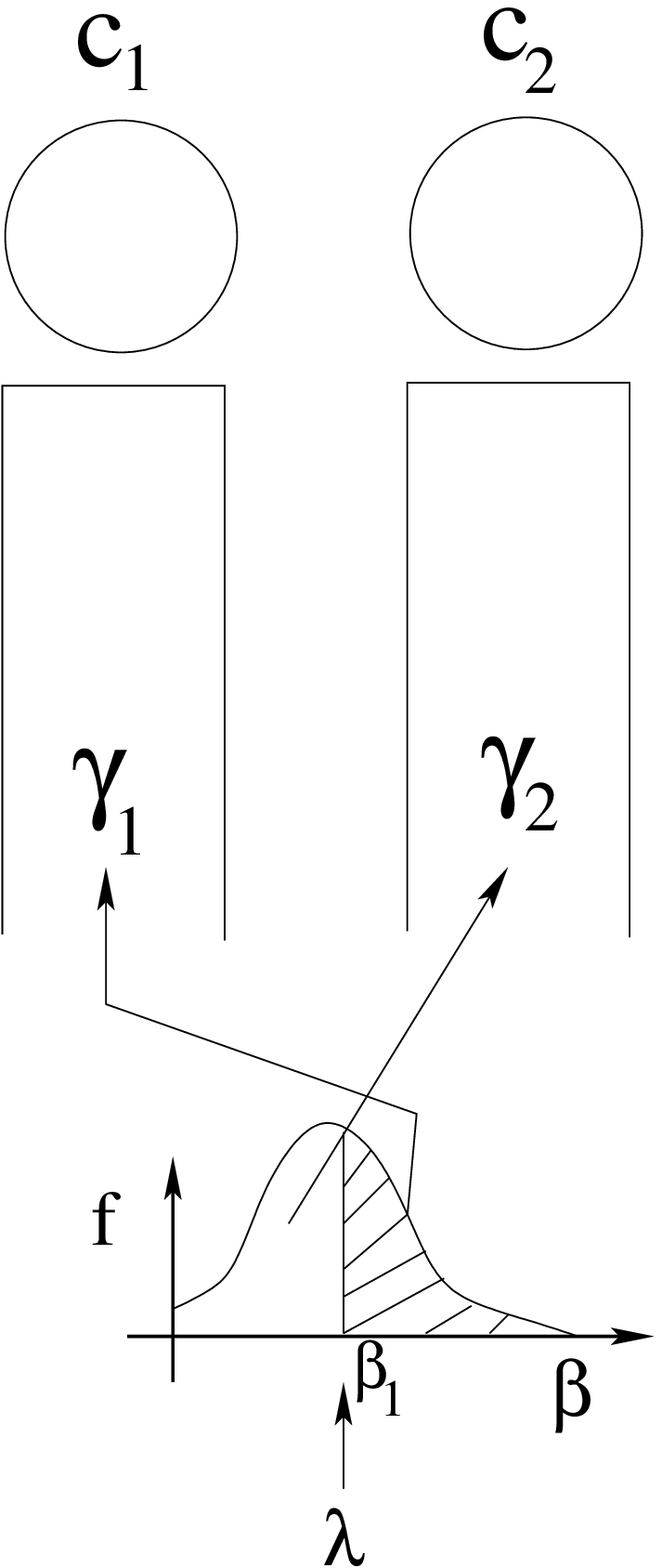}
      \caption{Representation of $K^W$ when $c_1 > c_2.$}  
      \label{fig:c1gc2}
    \end{center}
  \end{minipage}
  \hspace{0.8cm}
  \begin{minipage}{6cm}
    \begin{center}
      \includegraphics[width=.45\columnwidth, height=6cm]{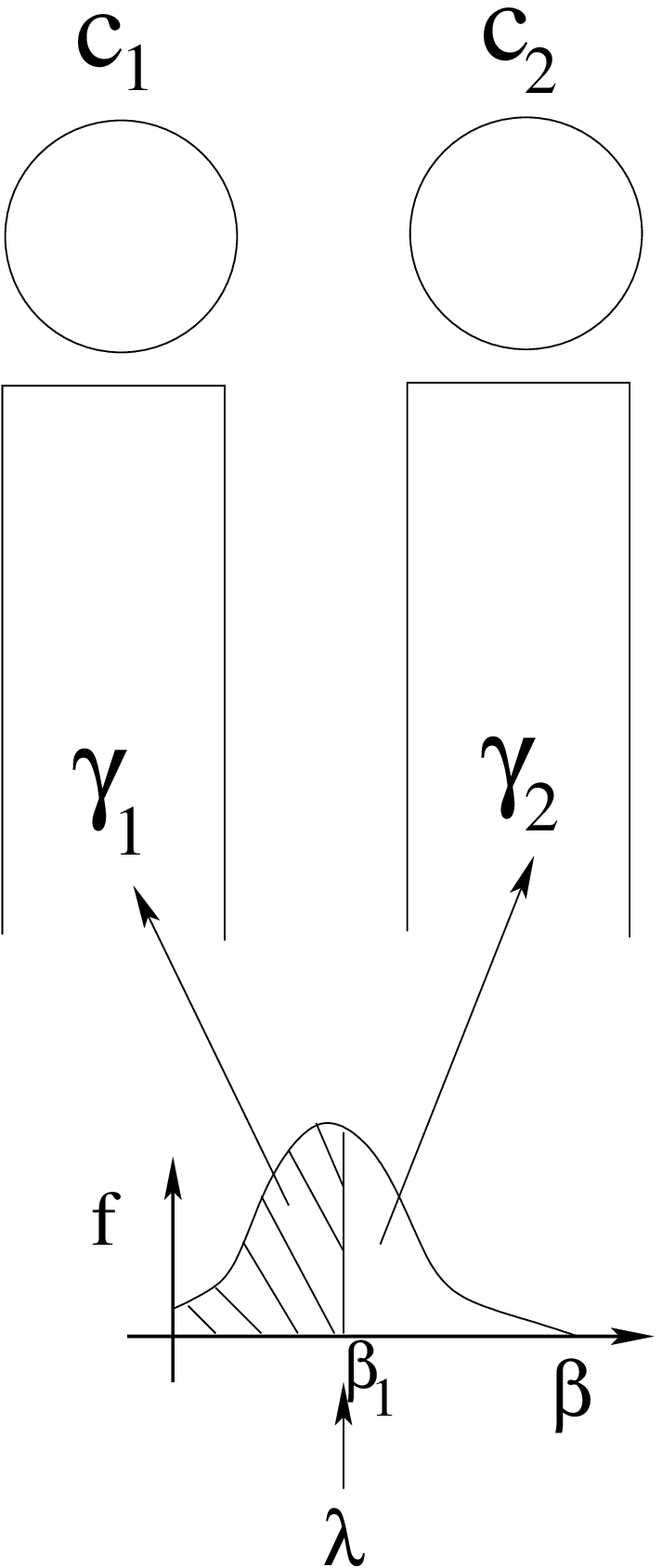}
      \caption{Representation of $K^W$ when $c_1 < c_2.$}
      \label{fig:c1lc2}
    \end{center}
  \end{minipage}
\end{figure*} 
\begin{theorem} 
  \label{thm:wardrop-cts} 
  Define $\delta_i$ as the probability distribution that puts unit mass
  on $i$ and suppose that the kernel $K^W$ satisfies the Wardrop equilibrium
  condition. Then there exists a threshold $\beta_1$ with
  $\beta_1 \in [a,b]$ such that
\begin{itemize}
 \item when $c_1 > c_2~ (\mbox{resp.~} c_1 < c_2),$ 
\begin{eqnarray}
\label{eq:wardrop-cts}
   K^W(\beta,\cdot) & =  & \begin{cases} 
      \delta_1~ (\mbox{resp.~} \delta_2)  & \mbox{ for } \beta \in (\beta_1,b],\\
      \delta_2~ (\mbox{resp.~} \delta_1) & \mbox{ for } \beta \in [a,\beta_1].
          \end{cases} 
\end{eqnarray}
Further if $\beta_1 \in (a,b)$ then, 
\begin{equation}
\label{eq:beta_1}
c_1 + \beta_1 D_1(\gamma_1) = c_2 + \beta_1 D_2 (\gamma_2).
\end{equation}

\item When $c_1 = c_2,$ $K^W$ is not unique and any kernel 
${K}$ with $\gamma_1 = \gamma^+$ is a valid Wardrop equilibrium 
kernel $K^W$ where $\gamma^+:= \left\lbrace \gamma_1 : D_1(\gamma_1) = D_2(\gamma_2)
\right\rbrace.$ 
\end{itemize}
\end{theorem}
Refer Figures \ref{fig:c1gc2} and \ref{fig:c1lc2} for a representation of the 
Wardrop equilibrium kernel for the case when $c_1 > c_2$ and $c_1 < c_2$ respectively. 
Here $f(\cdot)$ denotes the underlying density function of the random variable 
$\bm{\beta}$ while the shaded region identifies the delay cost parameter of those
 customers that choose Server~1. 

 \begin{proof}
The first part is simply a restatement of 
Corollary 4 in \cite{Bodas14} for the case $c_1 > c_2$ and the proof
for $c_1< c_2$ is along similar lines. We now prove the second part.
Consider the case when $c_1 = c_2$ and recall the assumption
that $D_1(0) < D_2(\lambda)$ and $D_2(0) < D_1(\lambda).$
$K^W$ must be such that $D_1(\gamma_1) = D_2(\gamma_2).$ 
To see why this must be true, suppose that this is not true and let
$D_1(\gamma_1) \neq D_2(\gamma_2).$ Customers from the queue with
a higher delay cost will have an incentive to move to the queue 
with a lower delay cost. This implies that a $K^W$ with $D_1(\gamma_1) 
\neq D_2(\gamma_2)$ is not at equilibrium. Recall the definition
$\gamma^+:= \left\lbrace \gamma_1 : D_1(\gamma_1) = D_2(\gamma_2)
\right\rbrace.$ Since, $D_1(0) < D_2(\lambda)$ and $D_2(0) < D_1(\lambda),$
we have $0 < \gamma^+ < \lambda.$ Now for any kernel $K$ satisfying
$\gamma_1 = \gamma^+,$ since $c_1 = c_2,$ the cost for any customer
at the two servers is equal. Hence there is no incentive for any customer to 
deviate from its choice of the server. The Wardrop equilibrium kernel
$K^W$ though not unique must however satisfy  
$\lambda \int_{\beta=a}^{b} K^W(\beta,j) dF(\beta) = \gamma^+ $.
\end{proof}

\section{Problem Formulations}
\label{sec:formulation}
Having characterized the Wardrop equilibrium kernel $K^W$ for a two server
system, we will now formulate the revenue maximization problems for both the 
monopoly and the duopoly model. Let $R_j(c_j, \gamma_j):= c_j \gamma_j$ 
denote the revenue rate at server $j$
when the arrival rate of customers due to the corresponding kernel $K^W$ is $\gamma_j$
for $j = 1, 2.$  For the monopoly model, let $R_T(c_1, \gamma_1)$
denote the revenue rate for the monopoly service system.
Since $\gamma_2 = \lambda - \gamma_1,$ it suffices 
to express the revenue rate as a function of only $\gamma_1.$
We have 
$$	
R_T(c_1, \gamma_1) := c_1\gamma_1 + c_2 \gamma_2 = c_2\lambda + 
(c_1 - c_2)\gamma_1.
$$
Note from Theorem \ref{thm:wardrop-cts}, that the argument $\gamma_1$ is determined by
  the kernel $K^W$ which in turn depends on the admission prices $c_1$ 
  and $c_2.$ This dependence will be made explicit by writing $\gamma_1$ as
  $\gamma_1(c_1, c_2)$ and the revenue optimization problem for the monopoly can now be 
stated as follows.
  \begin{equation}
    \label{prog:rev_mono}
  \tag{P1}  
 \begin{aligned}
      & \underset{c_1}{\max}
      & & R_T(c_1, \gamma_1(c_1, c_2)) = c_2\lambda + 
(c_1 - c_2)\gamma_1(c_1, c_2) \\
      & \text{subject to} & &  0 \leq c_1 \leq c^1
    \end{aligned}
  \end{equation}
  where $c^1$ is an arbitrarily large value such that $\gamma_1(c^1,c_2) = 0.$
    $c^1$ is a technical requirement to ensure a compact domain and
  one could also define $c^1:= \inf \left\lbrace c: \gamma_1(c,c_2) = 0
  \right\rbrace$ in which case we have $\gamma_1(c_1,c_2) = 0$
  for any $c_1 > c^1.$
  To be able to solve program $\ref{prog:rev_mono}$
  using standard optimization techniques, a closed form 
  expression for $\gamma_1(c_1, c_2)$ would be convenient. When $c_1 > c_2, $ 
  and $\beta_1 \in (a,b),$ from Theorem \ref{thm:wardrop-cts} and the definition
  of $\gamma_1$, it can be seen that
  \begin{equation}
  \label{eq:gamma_c1_c2}
\gamma_1(c_1, c_2) = \lambda(1 - F(\beta_1))   
  \end{equation}
where 
    $$
  \beta_1 = \left\lbrace \beta : c_1 + \beta D_1(\lambda(1 - F(\beta)))
   = c_2 + \beta D_2(\lambda F(\beta))\right\rbrace.
  $$
  
  A similar condition follows when $c_1 <  c_2$ and it can be seen
  that obtaining an explicit expression for $\gamma_1(c_1, c_2)$ is difficult.
  Note that we have not assumed any functional form for $D_j$ and $F(\cdot)$
  and for certain choice of these functions, a closed form expression for  
  $\gamma_1(c_1, c_2)$ may not be possible. Without an analytic expression
  for $\gamma_1(c_1, c_2),$ it is difficult to solve the revenue maximization problem.
  Therefore we require an alternative approach to solve program 
  $\ref{prog:rev_mono}.$ One possible alternative is to let the equilibrium 
  $\gamma_1$ (the value of $\gamma_1$ at equilibrium) 
  be the optimization variable and represent other variables of the system 
  such as $c_1, c_2, \beta_1$ as a function of $\gamma_1.$
 With slight abuse of notation, we will use $c_j(\gamma_j)$ to denote
 the admission price at Server~$j$ when the  arrival rate to Server~$j$
 at equilibrium is $\gamma_j$ where $j = 1,2.$ Similarly, we shall use
 $\beta_1(\gamma_1)$ to represent the threshold $\beta_1$ corresponding 
 to an equilibrium arrival rate of $\gamma_1$
 to Server~$1.$  Note that $c_1(\gamma_1)$ is also a function of $c_2.$
 This is because the equilibrium $\gamma_1$ depends on the difference
 $(c_1 - c_2)$ and not on their individual values. This 
 is clear from Theorem \ref{thm:wardrop-cts} (Eq. \eqref{eq:beta_1}).
 Therefore for a given $c_2$ and $\gamma_1 \in (0, \lambda)$ one can determine $c_1$
 using Eq. \eqref{eq:beta_1}. We have suppressed this
 dependence on $c_2$ to simplify notation.
 For the monopoly model $c_2(\gamma_2)=c_2$ as $c_2$ is 
 assumed fixed. Thus the equivalent revenue optimization problem  
 for the monopoly is as follows.
  \begin{equation}
    \label{prog:rev_mono_new}
  \tag{P2}  
 \begin{aligned}
      & \underset{\gamma_1}{\max}
      & & R_T(c_1(\gamma_1), \gamma_1) = c_2\lambda + 
(c_1(\gamma_1) - c_2)\gamma_1 \\
      & \text{subject to} & &  0 \leq \gamma_1 \leq \gamma^1(c_2) \leq \lambda
    \end{aligned}
  \end{equation}
  where $\gamma^1(c_2)$ determines the domain for the feasible values 
  of $\gamma_1$ as a function of $c_2.$ An
  intuitive explanation for the quantity $\gamma^1(c_2)$ is as follows.
  Consider the case $c_1 = c_2 = 0.$ From Theorem \ref{thm:wardrop-cts},
  we have $\gamma_1 = \gamma^+$ where $0 < \gamma^+ < \lambda.$
  Using the notation $\gamma_1(c_1, c_2),$ we have 
  $\gamma_1(0, 0) = \gamma^+.$ For any $c_1 > 0,$  
  $\gamma_1(c_1, 0) < \gamma^+$ since the increase in the admission price 
  at Server~1 makes the server more costly and decreases the 
  resulting $\gamma_1.$
  Clearly, for any $c_1 \geq 0$ and $c_2 = 0,$ $\gamma_1 \notin (\gamma^+, \lambda]$ 
  and $c_1(\gamma_1)$ in program \ref{prog:rev_mono_new} cannot be defined
  for $\gamma_1 \in (\gamma^+, \lambda].$
  Therefore when $c_2 = 0,$ the domain for the optimization 
  variable $\gamma_1$ should be restricted to $[0, \gamma^+].$
  In general, for an arbitrary $c_2,$ the domain for 
  $\gamma_1$ in program \ref{prog:rev_mono_new} is defined 
  using $\gamma^1(c_2)$ and this will be characterized formally 
  in Section \ref{sec:monopoly}. 
    
Now consider the duopoly market with two competing servers charging
admission prices $c_1$ and $c_2$ to their arriving customers.
The objective of Server~$j$ is to choose an admission price $c_j$
that maximizes its revenue rate $R_j.$ For this duopoly,
the revenue optimization problem for Server~$j$ is as follows.

\begin{equation*}
   \label{prog:rev_duo}
   \tag{P3}  
 \begin{aligned}
      & \underset{c_j}{\max}
      & & R_j(c_j, \gamma_j) = c_j \gamma_j(c_j, c_{j^-}) \\
      & \text{subject to} & &  0 \leq c_j \leq c^j \\
      & \text{given} & & c_{j^-}
    \end{aligned}
\end{equation*} 

where $c_{j^-}$ represents the admission price at the server other than $j,$ i.e.,
$c_{1^-} = c_2$ and $c_{2^-} = c_1.$ 

For the duopoly market, the aim is to obtain the Nash equilibrium set of 
admission prices to be charged at the two servers. We shall denote the 
Nash equilibrium prices by the tuple $(c_1^*, c_2^*).$ Using the notion of 
the best response function \cite{Osborne03}, $(c_1^*, c_2^*)$ can be characterized
as follows. Let $B_i(c_{i^-})$ denote the admission price at 
Server~$i$ that maximizes the server revenue $R_i$ for a given value of $c_{i^-}$
for $i=1,2.$ Clearly, $B_i(c_{i^-})$ is the maximizer in program \ref{prog:rev_duo}
and it is easy to see that

\begin{eqnarray*}
 B_1(c_2) &:=& \left\lbrace c_1 \geq 0 : c_1 \gamma_1(c_1, c_2) \geq c_1^{\prime} \gamma_1(c^{\prime}_1, c_2) \forall c^{\prime}_1 \geq 0\right\rbrace \\
 B_2(c_1) &:=& \left\lbrace c_2 \geq 0 : c_2 \gamma_2(c_1, c_2) \geq c_2^{\prime} \gamma_2(c_1, c^{\prime}_2) \forall c^{\prime}_2 \geq 0\right\rbrace.
 \end{eqnarray*}
and 
\begin{eqnarray*}
  (c_1^*, c_2^*) = \left\lbrace (c_1, c_2) : B_1(c_2) = c_1, B_2(c_1) = c_2\right\rbrace.
 \end{eqnarray*} 
However as argued earlier, the closed form expression for $\gamma_j(c_j, c_{j^-})$ is not
easy to obtain. This makes it difficult to solve program \ref{prog:rev_duo} and obtain the 
best responses $B_i(c_{i^-})$ for $i=1,2.$ As a result, obtaining $(c_1^*, c_2^*)$ is in 
general not easy. As in the case of the monopoly program, to obtain $(c_1^*, c_2^*),$
we need to first reformulate program \ref{prog:rev_duo} by letting $\gamma_j$ denote the 
optimizing variable. The corresponding optimization problem is as follows.
  \begin{equation}
    \label{prog:rev_duo_new}
   \tag{P4}  
 \begin{aligned}
      & \underset{\gamma_j}{\max}
      & & R_j(c_j(\gamma_j), \gamma_j) := c_j(\gamma_j)\gamma_j \\
      & \text{subject to} & &  0 \leq \gamma_j \leq \gamma^j(c_{j^-}) \leq \lambda \\
      & \text{given} & & c_{j^-}.
    \end{aligned}
  \end{equation}  
  $c_j(\gamma_j)$ can be interpreted as the admission
  price at Server~$j$ that leads to the equilibrium arrival rate
  of $\gamma_j$ when the other server charges $c_{j^-}.$ 
  Note again that $c_j(\gamma_j)$ will be a 
  function of $c_{j^-}$ but we do not make this explicit in the notation. 
  To lighten notation, we will not make this dependence explicit.
  Now let $\gamma_1^*(c_2)$ denote the maximizer in program \ref{prog:rev_duo_new}
  for a given value of $c_2$. Then the best response $c_1$ is in fact given by the function
  $c_1(\gamma_1^*(c_2)).$ Therefore, once the function $c_1(\gamma_1)$ is characterized, 
  the best response now denoted by $\hat{B}_1(c_2)$ satisfies
  $\hat{B}_1(c_2) = c_1(\gamma_1^*(c_2)).$  We now have 
  \begin{eqnarray*}
  (c_1^*, c_2^*) = \left\lbrace (c_1, c_2) : \hat{B}_1(c_2) = c_1, \hat{B}_2(c_1) = c_2\right\rbrace
 \end{eqnarray*} 
  where $\hat{B}_i(c_{i^-}) = c_i(\gamma_i^*(c_{i^-}))$ and 
  as stated earlier, $\gamma_i^*(c_{i^-})$ is the maximizer in program \ref{prog:rev_duo_new}
  for $i=1,2.$ It is therefore clear that $(c_1^*, c_2^*)$ can be 
  obtained  once we have characterized $c_1(\gamma_1).$
  We shall analyze the program  $\ref{prog:rev_duo_new}$
  in detail in Section \ref{sec:duopoly} and explicitly
  characterize the functions $c_j(\gamma_j)$ for $j = 1,2$ 
  to be able to obtain $(c_1^*, c_2^*)$.

%

\section{Monopoly Market}
\label{sec:monopoly}
In this section, we will analyze the monopoly program $\ref{prog:rev_mono_new}.$
To be able to solve program $\ref{prog:rev_mono_new},$ we need to
characterize $c_1(\gamma_1)$ for a fixed value of $c_2.$ This procedure is outlined below.
From Eq. \eqref{eq:beta_1} of Theorem \ref{thm:wardrop-cts}, we know that
when $\beta_1 \in (a,b),$ (and hence $\gamma_1 \in (0, \lambda)$) we have
$$c_1 - c_2 = \beta_1 \left( D_2(\gamma_2) - D_1(\gamma_1)\right).$$
We will express the right hand side of the above equation
as a function of $\gamma_1,$ i.e.,  
\begin{equation}
\label{eq:ggamma}
 g_1(\gamma_1) :=  \beta_1(\gamma_1)\left(D_2(\lambda - \gamma_1) - D_1(\gamma_1)\right)
\end{equation}
where  $\beta_1(\gamma_1)$ represents the threshold $\beta_1$
for a kernel $K^W$ that satisfies Theorem \ref{thm:wardrop-cts} and  
corresponds to an equilibrium arrival rate of $\gamma_1.$
Note that $g_1(\gamma_1)$ characterizes the difference
$(c_1 - c_2)$ as a function of $\gamma_1.$
For a fixed $c_2$ and for a $\gamma_1$ satisfying 
$0 \leq \gamma_1 \leq \gamma^1(c_2) \leq \lambda$ (the 
domain of $\gamma_1$ in program \ref{prog:rev_mono_new})
we see that $c_1(\gamma_1) = c_2 + g_1(\gamma_1).$
We characterize $c_1(\gamma_1)$ in the following manner. 
We first characterize $\beta_1(\gamma_1)$ using Lemmas 
\ref{lemma:compare_gamma_c} and \ref{lemma:betagamma}. Then in Lemma
\ref{lemma:gfx}, we characterize $g_1(\gamma_1).$ For a fixed $c_2,$
we then obtain  $\gamma^1(c_2)$ in Lemma \ref{lemma:gamma^1} that 
determines the domain of $c(\gamma_1)$ .
To prove this lemma we need to characterize the uniqueness of kernel $K^W$
for a fixed difference $(c_1 - c_2).$  This is part of 
Lemma \ref{lemma:threshold}. Finally we characterize $c_1(\gamma_1)$ 
in Theorem \ref{thm:cgamma} using $g_1(\gamma_1)$ and 
$\gamma^1(c_2).$

Recall that we make minimal assumptions on the distribution 
$F(\cdot)$ and on the delay cost function $D_j(\cdot).$ 
For our numerical examples and also to illustrate the properties
of the functions $\beta_1(\cdot), g_1(\cdot)$ and $c_1(\cdot),$ 
we consider the following examples for $F(\cdot)$ and $D_j(\cdot).$
The distribution $F(\cdot)$ is from one of the following;
\begin{itemize}
\item Uniform distribution over the range $[a,b].$ 
\item Exponential distribution with mean $\tau.$
\item Gamma distribution with shape $k$ and scale $\theta$. 
\end{itemize}
For the delay cost function, we shall assume one of the following.
\begin{itemize}
\item $D_j(\gamma_j) = \frac{\gamma_j}{\mu_j}.$ This corresponds to the 
case of linear delay.
\item  $D_j(\gamma_j) = \frac{1}{\mu_j - \gamma_j}$ and $\mu_j > \lambda.$
 This corresponds to $M/M/1$ type delay cost function.
\end{itemize}
The distribution and the delay cost functions outlined above are commonly used to 
model heterogeneous customers and congestion costs.
(Refer \cite{Bodas11a,Bodas11b,Bodas14,Luski76,Levhari78})

We now begin with the following lemma that identifies the necessary and sufficient 
condition on the equilibrium $\gamma_1$ when either $c_1 \geq c_2$ or $c_1 < c_2.$ 
\begin{lemma}
 \label{lemma:compare_gamma_c}
 $\gamma_1 \in [0,\gamma^+]$ iff $c_1 \geq c_2$ while 
  $ \gamma_1 \in (\gamma^+,\lambda]$ iff $c_1 < c_2.$
 \end{lemma}
  \begin{proof}
 See Appendix for proof.
\end{proof}
Refer Figure \ref{fig:lemma1_1} and \ref{fig:lemma1_2} for an illustration of the 
lemma.
%
%

\begin{figure*}
  \begin{minipage}{6cm}
    \begin{center}
      \includegraphics[width=.8\columnwidth, height=6cm]{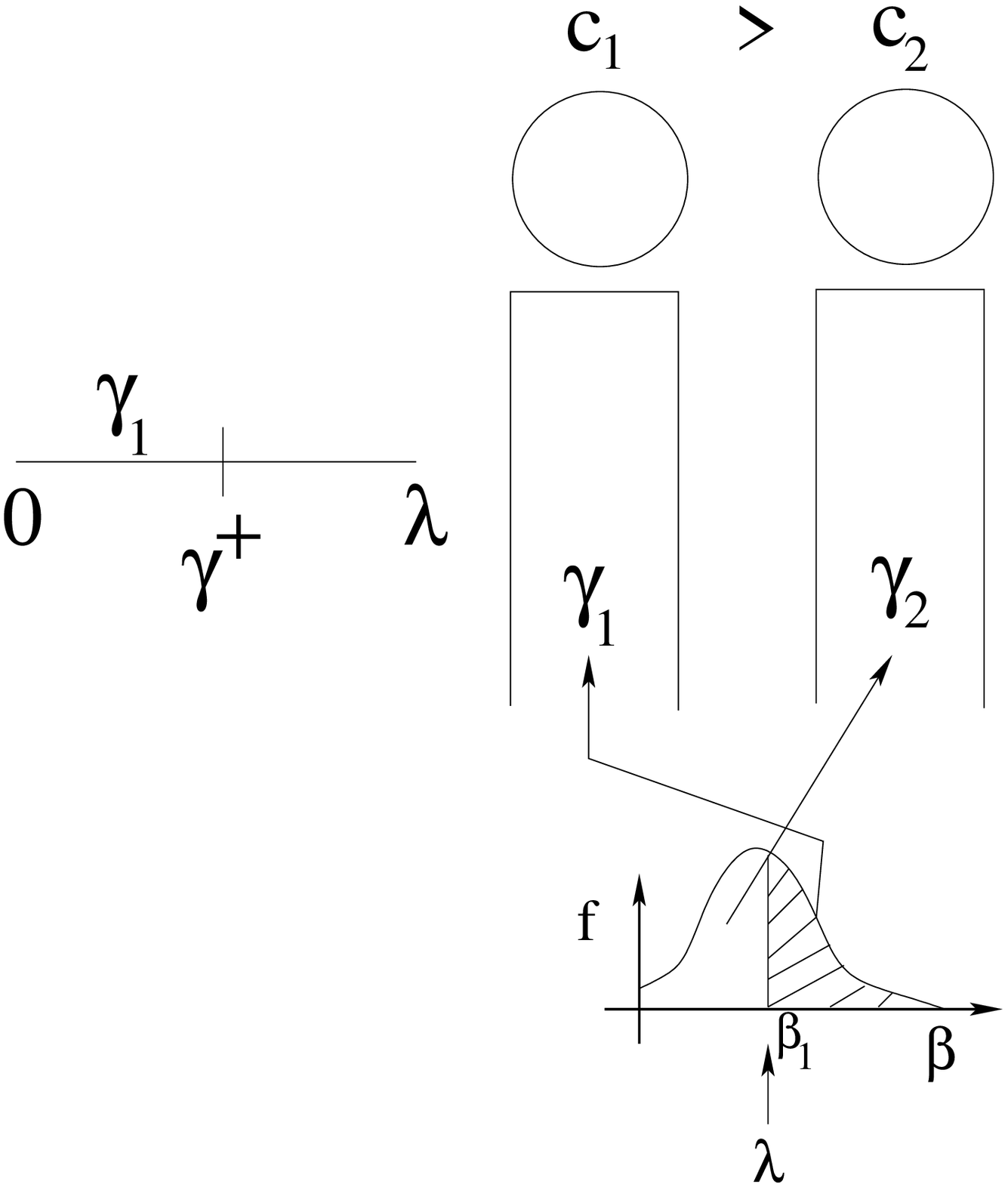}
       \caption{Illustrating $\gamma_1 \in [0,\gamma^+]$ and $c_1 \geq c_2.$}
\label{fig:lemma1_1}
    \end{center}
  \end{minipage}
  \hspace{0.8cm}
  \begin{minipage}{6cm}
    \begin{center}
      \includegraphics[width=.8\columnwidth, height=6cm]{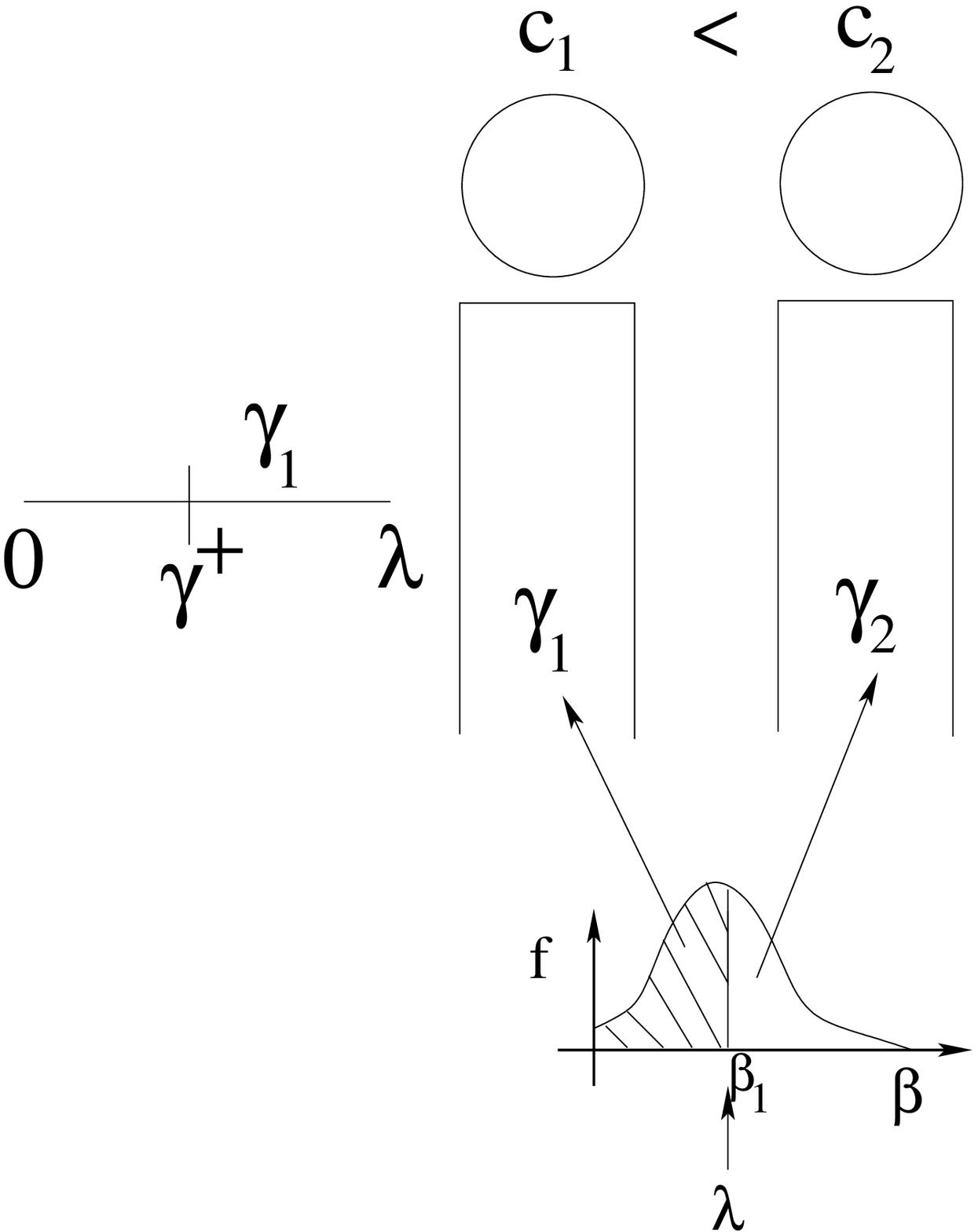}
      \caption{Illustrating $ \gamma_1 \in (\gamma^+,\lambda]$ and $c_1 < c_2.$}
   \label{fig:lemma1_2}
    \end{center}
  \end{minipage}
\end{figure*}

Next, we express the threshold
$\beta_1$ of Theorem \ref{thm:wardrop-cts} as a function of $\gamma_1$.
Recall from the theorem that $K^W$ is characterized by $\beta_1$ when 
$c_1 \neq c_2.$  We let $\beta_1(\gamma_1)$ to
denote the value of the threshold $\beta_1$ (characterizing $K^W$) 
for a given $\gamma_1$ such that $\gamma_1 \neq \gamma^+.$ 
We have the following lemma.
\begin{lemma}
\label{lemma:betagamma}
 \begin{eqnarray}
    \beta_1(\gamma_1) & = & \begin{cases} 
       F^{-1}\left(\frac{\lambda - \gamma_1}{\lambda}\right) & 
       \mbox{ for } 0 \leq \gamma_1 < \gamma^+,\\
      F^{-1} \left(\frac{\gamma_1}{\lambda}\right) & \mbox{ for } 
      \gamma^+ < \gamma_1 \leq  \lambda.
          \end{cases} 
\end{eqnarray}
where $F^{-1}$ represents the quantile function or the inverse 
function of the distribution $F.$ 
\end{lemma}
 \begin{proof}
 See Appendix for proof.
\end{proof}

Note that $\beta_1(\gamma_1)$ is not defined in Lemma
\ref{lemma:betagamma} when $\gamma_1 = \gamma^+.$
This is because the Wardrop kernel $K^W$ with
$\gamma_1 = \gamma^+$ is not unique and need not be characterized by 
a single threshold. We shall however assume from now on that when 
$\gamma_1 = \gamma^+$ (and hence $c_1 = c_2$), the corresponding kernel
$K^W$ is also characterized by a single threshold $\beta_1.$ Hence for 
$c_1 = c_2,$ we have
\begin{eqnarray}
   K^W(\beta,1) & =  & \begin{cases} 
      \delta_1  & \mbox{ for } \beta \in (\beta_1,b],\\
      \delta_2  & \mbox{ for } \beta \in [a,\beta_1].
          \end{cases} 
\end{eqnarray}

As a result, we define $\beta_1(\gamma^+)
= F^{-1}\left(\frac{\lambda - \gamma^+}{\lambda}\right)$
and the modified $\beta_1(\gamma_1)$ is now as
follows.
\begin{eqnarray}
\label{eq:betagamma}
    \beta_1(\gamma_1) & = & \begin{cases} 
       F^{-1}\left(\frac{\lambda - \gamma_1}{\lambda}\right) & 
       \mbox{ for } 0 \leq \gamma_1 \leq \gamma^+,\\
      F^{-1} \left(\frac{\gamma_1}{\lambda}\right) & \mbox{ for } 
      \gamma^+ < \gamma_1 \leq  \lambda.
          \end{cases} 
\end{eqnarray}

\begin{figure}
  \begin{center}
    \includegraphics[height=2in]{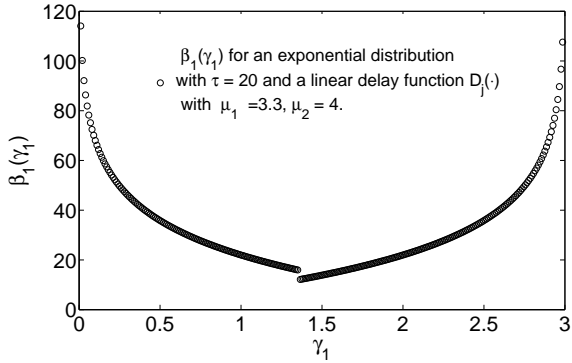}
  \end{center}
  \caption{Illustrating $\beta_1(\gamma_1)$ when the servers are not identical.}
  \label{fig:betagamma}
\end{figure}

Refer Fig. \ref{fig:betagamma} for a numerical evaluation of 
Eq. \eqref{eq:betagamma} for the case when $F(\cdot)$ is an
exponential distribution with $\tau = 20.$ The delay functions 
are $D_j(\gamma_j) = \frac{\gamma_j}{\mu_j}$ where $\mu_1 = 3.3$
and $\mu_2 = 4.$ Fig. \ref{fig:betagamma2}
corresponds to the case when the two servers are identical, i.e., 
$\mu_1 = \mu_2 = 4.$

\begin{remark}
 Recall our assumption that $F(\cdot)$ is absolutely 
 continuous and strictly increasing in its domain. Further,
 the support is $[a,b]$ and hence $F(\cdot)$ is a bijective
 function whose inverse exists. In fact $F^{-1}(\cdot)$ is continuous
 and strictly increasing in its domain.  Since $F^{-1}(\cdot)$ is 
 continuous in its arguments, $\beta_1(\gamma_1)$
 is continuous when $0 \leq \gamma_1 < \gamma^+$ and $\gamma^+ < \gamma_1 \leq \lambda.$
 However at $\gamma_1 = \gamma^+,$  $\beta_1(\gamma_1)$ is in general not continuous (Refer Fig. \ref{fig:betagamma}).
 For the case when the servers are identical, i.e., $D_1(\gamma) = D_2(\gamma) = D(\gamma),$ 
 we see from the definition of $\gamma^+$ that $\gamma^+  = \frac{\lambda}{2}.$
 For this case, it is easy to see that $\beta_1(\gamma_1)$  
 is continuous at $\gamma_1 = \gamma^+,$ (but not differentiable).
 See Fig. \ref{fig:betagamma2}.    
 \end{remark}
\begin{figure}
  \begin{center}
    \includegraphics[height=2in]{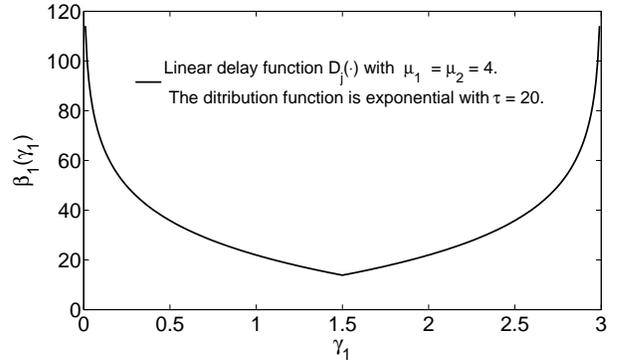}
  \end{center}
  \caption{Illustrating $\beta_1(\gamma_1)$ for the case of identical servers.}
  \label{fig:betagamma2}
\end{figure}


Having obtained $\beta_1(\gamma_1),$ we shall now analyze $g_1(\gamma_1)$ 
that was defined in Eq. \eqref{eq:ggamma}.
$g_1(\gamma_1)$ will be used later to obtain $c_1(\gamma_1).$
We have the following lemma.
\begin{lemma}
\label{lemma:gfx}
For $0 \leq \gamma_1 \leq \lambda,$ $g_1(\gamma_1)$ is continuous and monotonic 
decreasing in $\gamma_1.$ Further, $g_1(\gamma^+) = 0.$ 
\end{lemma}

 \begin{proof}
 See Appendix for proof.
\end{proof}

See Fig. \ref{fig:ggamma} for a numerical evaluation of $g_1(\gamma_1)$ when 
$F(\cdot)$ is an exponential distribution with $\tau = 20$ and when the servers
have a linear delay with $\mu_1 = 3.3$ and $\mu_2 = 4.$
\begin{figure}
  \begin{center}
    \includegraphics[height=2.5in]{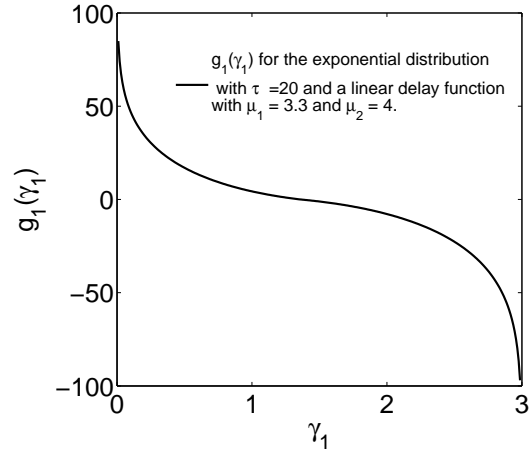}
  \end{center}
  \caption{Illustrating $g_1(\gamma_1)$ when the servers are not identical.}
  \label{fig:ggamma}
\end{figure}

To determine $c_1(\gamma_1),$ we also need to identify the domain over which 
it can be defined. As argued earlier, this domain is determined by 
$\gamma^1(c_2)$ which is characterized in Lemma \ref{lemma:gamma^1}. 
Before stating Lemma \ref{lemma:gamma^1}, we shall first characterize 
the uniqueness of $\beta_1,$ and hence the kernel $K^W,$ when $c_1 \neq c_2.$ 
While Theorem \ref{thm:wardrop-cts}, guarantees existence of a $\beta_1$ 
characterizing kernel $K^W,$ it does not guarantee the uniqueness of 
$\beta_1$ and hence the uniqueness of the kernel $K^W$. This result will 
be used in the proof of Lemma \ref{lemma:gamma^1}.

\begin{lemma}
\label{lemma:threshold}
For a given $\Delta := (c_1 - c_2),$ the threshold $\beta_1$ characterizing the kernel
$K^W$ in Theorem \ref{thm:wardrop-cts} is as follows.
\begin{eqnarray}
    \beta_1 & = & \begin{cases} 
       b & \mbox{ ~if~ } \Delta \geq g_1(0)  \mbox{ ~or~ } \Delta \leq g_1(\lambda),\\
       \beta_1(\hat{\gamma})  & \mbox{~if~} g_1(\lambda) < \Delta < g_1(0) 
       \end{cases}      
\end{eqnarray}
where $\hat{\gamma}$ satisfies $\Delta = g_1(\hat{\gamma}).$
For a fixed $\Delta,$ the equilibrium $\gamma_1$ and the corresponding 
$\beta_1$ is unique and this implies the uniqueness of $K^W.$
\end{lemma}
 \begin{proof}
 See Appendix for proof.
\end{proof}

We now characterize $\gamma^1(c_2)$ in the following lemma.

\begin{lemma}
\label{lemma:gamma^1}

\begin{eqnarray}
    \gamma^1(c_2) & = & \begin{cases} 
       \lambda & \mbox{ for } c_2 \geq - g_1(\lambda),\\
       \gamma : g_1(\gamma) = -c_2 & 
\mbox{ for } c_2 < - g_1(\lambda) \end{cases} 
\end{eqnarray}

In words, when $c_2<-g_1(\lambda)$ we have $\gamma^1(c_2) = 
\left\lbrace \gamma : g_1(\gamma) = -c_2\right\rbrace$ and for any $c_1 \geq 0,$
the equilibrium $\gamma_1 \notin (\gamma^1(c_2),\lambda].$
However when $c_2 \geq - g_1(\lambda),$ we have $\gamma^1(c_2) = \lambda$ 
in which case for suitable choices of $c_1,$ $\gamma_1 \in [0,\lambda].$
\end{lemma}

 \begin{proof}
 See Appendix for proof.
\end{proof}

The above lemma also implies that, 
if $c_{2} \geq -g_1(\lambda),$
then for any $c_1 \in (0,c_2 +g_1(\lambda))$ the equilibrium 
$\gamma_1$ satisfies $\gamma_1 = \lambda.$ On the other hand,
if the parameters 
of the system are such that $c_2 + g_1(\lambda) < 0,$ then 
for any set of admission prices $c_1$ at Server~1, we have
$\gamma_1 <  \gamma^1(c_2).$

\begin{remark}
\label{rem:g1+}
From Lemma \ref{lemma:gamma^1}, when $c_2 < - g_1(\lambda),$ 
we have $\gamma^1(c_2) = \gamma$ where $g_1(\gamma) = -c_2.$
From Lemma \ref{lemma:gfx}, we know that $g_1(\gamma) < 0$ for 
$\gamma > \gamma^+.$ Hence when $c_2 \geq 0,$ we have  
$\gamma^1(c_2) \geq \gamma^+$ with strict equality when 
$c_2 = 0.$ 
\end{remark}

Finally, we have the following theorem to express
$c_1$ as a function of $\gamma_1,$ denoted by $c_1(\gamma_1).$ 

\begin{theorem}
\label{thm:cgamma}
$c_1(\gamma_1) = c_2 + g_1(\gamma_1)$ for $0 < \gamma_1 < \gamma^1(c_2) \leq \lambda.$
For $\gamma_1 = 0,$  $c_1(0)$ must be at least equal to $c_2 + g_1(0),$
i.e., $c_1(0) \geq c_2 + g_1(0).$ Similarly when $\gamma^1(c_2) = \lambda,$
we have $c_1(\lambda) \leq c_2 + g_1(\lambda).$
\end{theorem}

\begin{proof}
First consider a fixed $\gamma_1$ satisfying 
 $\gamma_1 \in (0,\gamma^1(c_2))$ for a fixed $c_2.$
 ($\gamma^1(c_2)$ was characterized in Lemma \ref{lemma:gamma^1}.)
 The corresponding threshold $\beta_1$ is determined by 
 Eq. \eqref{eq:betagamma} and hence  we have 
 $\beta_1 \in (a,b)$ for  $\gamma_1 \in (0,\gamma^1(c_2)).$
Recall that Lemma \ref{lemma:threshold} relates the 
threshold $\beta_1$ with $\Delta.$ Since $\beta_1 < b,$
 from Lemma \ref{lemma:threshold}, $\Delta$ must satisfy
 $\Delta = g_1(\gamma_1).$ Therefore for a fixed $c_2,$
 the admission price $c_1(\gamma_1)$ resulting in the
 arrival rate of $\gamma_1$ at Server~1 is given by 
 \begin{equation*}
  c_1(\gamma_1) = c_2 + g_1(\gamma_1).
   \end{equation*}

   For the case $\gamma_1 = 0,$ from Eq. \eqref{eq:betagamma}, we have 
   $\beta_1 = b.$ From Lemma \ref{lemma:threshold}, this implies that
   $\Delta \geq g_1(0).$ From the definition of $\Delta,$ we have 
   $c_1(0) \geq c_2 + g_1(0).$
   Similarly when $\gamma_1 = \lambda,$ from Eq. \eqref{eq:betagamma},
   we have $\beta_1 = b.$ From Lemma \ref{lemma:threshold}, this implies 
   $\Delta \leq g_1(\lambda)$ and hence $c_1(\lambda) \leq c_2 + g_1(\lambda).$
   This completes the proof.
\end{proof}
 
In the above theorem, $c_1(0)$ and $c_1(\lambda)$ are not uniquely 
defined and can take values that satisfy  $c_1(0) \geq c_2 + g_1(0)$
and $c_1(\lambda) \leq c_2 + g_1(\lambda)$ respectively.
As convention, we henceforth define $c_1(0) = c_2 + g_1(0)$
and $c_1(\lambda) = c_2 + g_1(\lambda).$ Further note that 
the domain for $c_1(\cdot)$ is $0 \leq \gamma_1 \leq \gamma^1(c_2)$
and for $\gamma^1(c_2) <\gamma_1 < \lambda,$ $c_1(\gamma_1)$ is 
undefined. The function $c_1(\gamma_1)$
for $0 \leq \gamma_1 \leq \gamma^1(c_2) \leq \lambda$ 
can now be expressed as follows.

\begin{eqnarray}
\label{eq:cgamma}
    c_1(\gamma_1) & = & \begin{cases} 
       c_2 + g_1(\gamma_1) & \mbox{ for } 0 < \gamma_1 \leq \gamma^1(c_2) < \lambda,\\
       c_2 + g_1(0) & \mbox{ for } \gamma_1 = 0, \\
       c_2 + g_1(\lambda) & \mbox{ for }   \gamma_1 = \gamma^1(c_2) = \lambda,
\end{cases} 
\end{eqnarray}

 Now recall the revenue maximization problem $\ref{prog:rev_mono_new}.$
 Define $\gamma^*_1$ as the optimizer for this program with the
 revenue maximizing admission price given by $c_1(\gamma^*_1).$
 Since $R_T (c_1(\gamma_1),\gamma_1)= c_2\lambda +  (c_1(\gamma_1) 
 - c_2)\gamma_1,$ $\gamma^*_1$ must be such that $c_1(\gamma^*_1) > c_2.$ 
From Eq. \eqref{eq:cgamma}, this implies that $g_1(\gamma_1^*) > 0.$
From Lemma \ref{lemma:gfx} we have $g_1(\gamma_1) > 0$ for
$\gamma_1 \in (0,\gamma^+)$ and this implies that 
$\gamma_1^* \in (0, \gamma^+).$ The term $c_2\lambda$
in $R_T (c_1(\gamma_1),\gamma_1)$ is a constant and hence we have  
the following equivalent program for the revenue maximization problem.
 
\begin{equation}
    \label{prog:equi_mono}
  \tag{P4}  
 \begin{aligned}
      & \underset{\gamma_1}{\max}
      & & g_1(\gamma_1)\gamma_1 \\
      & \text{subject to} & &  0 \leq \gamma_1 \leq \gamma^+
    \end{aligned}
  \end{equation}
 where $g_1(\gamma_1)$ is given by Eq. \eqref{eq:ggamma}.
 
Note from Lemma \ref{lemma:gfx} that $g_1(\cdot)$
is a continuous function of its domain. Program \ref{prog:equi_mono}
involves maximizing a continuous function over a compact
set and hence a maximizer $\gamma_1^*$ exists.
The original monopoly program \ref{prog:rev_mono} has been
significantly simplified to the equivalent program
\ref{prog:equi_mono}. Since $g_1(\gamma_1)$ is strictly 
decreasing (and hence quasi-convex),
$g_1(\gamma_1)\gamma_1$ is in fact a product of two quasi-convex
functions. (However product of quasi-convex functions need not be
quasi-convex function). One can now use standard non-linear
optimization techniques to obtain $\gamma_1^*.$
To further understand Program \ref{prog:equi_mono}, we perform a 
numerical evaluation of $g_1(\gamma_1)\gamma_1 $ 
under a combination of assumptions on the distribution 
functions $F$ and the delay functions $D_j(\gamma_j)$ that 
were outlined earlier.

\begin{figure}
  \begin{center}
    \includegraphics[height=2in]{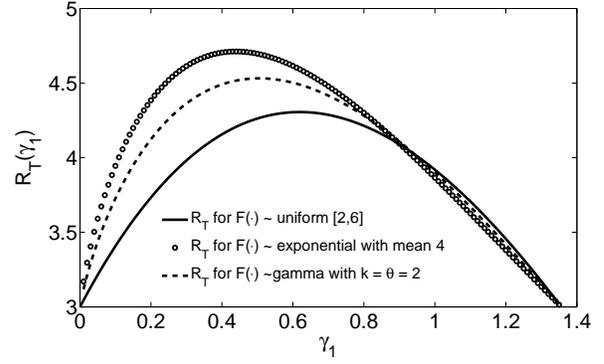}
  \end{center}
  \caption{$R_T$ as a function of $\gamma_1$ when $D_j(\gamma_j) = \frac{\gamma_j}{\mu_j}$}
  \label{fig:ex1}
\end{figure}

\textbf{Example 1:} In this example we shall assume that the 
$D_j(\gamma_j) = \frac{\gamma_j}{\mu_j}$ for $j = 1,2.$ We 
assume that $\mu_1 = 3.3$ and $\mu_2 = 4.$ Further, the 
arrival rate $\lambda = 3$ and we consider the 
following three examples for the distribution $F(\cdot).$ 
(1) $F$ has a uniform distribution with support on $[2,6].$
(2) $F$ has an exponential distribution with mean $\tau = 4$ and  
(3) $F$ has a Gamma distribution with the scale $k$ and shape $\theta$
parameters $2$ and $2$ respectively. Note that $\bm{\beta}$ with these three 
distributions have the same mean. We plot 
$R_T(c_1(\gamma_1),\gamma_1) = c_2 \lambda + g_1(\gamma_1)\gamma_1$
as a function of $\gamma_1$ in Fig.~\ref{fig:ex1} 
where we assume $c_2 = 1.$ When $F$ has the uniform distribution, 
$\gamma_1^* = 0.62.$ The optimal revenue rate  $R_T(\gamma_1^*) = 4.306$ while
the admission price $c_1(\gamma_1^*)$ maximizing $R_T$ is 3.106. 
The corresponding values for the exponential distribution are 
$\gamma_1^* = 0.44, R_T(\gamma_1^*) = 4.712$ and $c_1(\gamma_1^*) = 4.89$
while the values for gamma distribution are 
$\gamma_1^* = 0.51, R_T(\gamma_1^*) = 4.532$ and $c_1(\gamma_1^*) = 4$.

\textbf{Example 2:} In this example, we assume that 
$D_j(\gamma_j) = \frac{1}{\mu_j - \gamma_j}$ where again 
$\mu_1 = 3.3$ and $\mu_2 = 4.$ Note that $\lambda < \mu_j$
for $j = 1,2.$ The choice of $F(\cdot)$ is as in the previous 
example. A plot of $R_T(c_1(\gamma_1),\gamma_1)$ as a function 
of $\gamma_1$ is provided in Fig. \ref{fig:ex2}.  
When $F$ has the uniform distribution, 
$\gamma_1^* = 0.48.$ The optimal revenue rate  $R_T(\gamma_1^*) = 3.83$ while
the admission price $c_1(\gamma_1^*)$ maximizing $R_T$ is $2.72$. 
The corresponding values for the exponential distribution are 
$\gamma_1^* = 0.33, R_T(\gamma_1^*) = 4.21$ and $c_1(\gamma_1^*) = 4.67$
while the values for gamma distribution are 
$\gamma_1^* = 0.38, R_T(\gamma_1^*) = 4.04$ and $c_1(\gamma_1^*) = 3.74$.

\begin{figure}
  \begin{center}
    \includegraphics[height=2in]{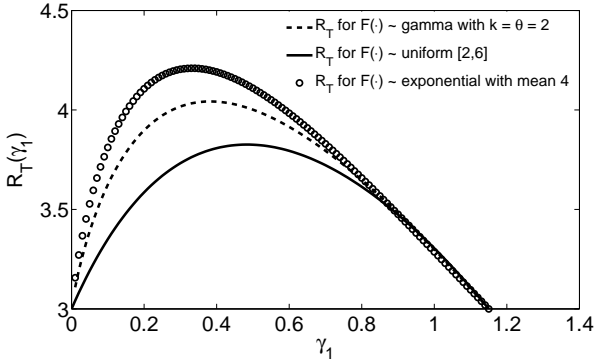}
  \end{center}
  \caption{$R_T$ as a function of $\gamma_1$ when $D_j(\gamma_j) = \frac{1}{\mu_j - \gamma_j}$}
  \label{fig:ex2}
\end{figure}

We conclude the analysis of the revenue maximization problem 
with the following observations made from the two examples given
above.
\begin{itemize}
 \item Firstly, we see that for the given examples of $F,$  
$R_T(c_1(\gamma_1),\gamma_1)$ is a unimodal function
in $\gamma_1.$ For the three distribution functions, 
it can be shown that $F^{-1}(\cdot)$ is differentiable 
in its arguments. For such distribution functions with 
differentiable $F^{-1}(\cdot),$ this implies that $g_1(\gamma_1)$
and hence $R_T(c_1(\gamma_1),\gamma_1)$ is differentiable in
$\gamma_1$ when $0 < \gamma_1 <\gamma^+.$ From 
Rolle's Theorem (Theorem 10.2.7 \cite{Tao06}), this 
implies that there exists a $\gamma_1 \in (0,\gamma^+)$
such that $\frac{dR_T}{d\gamma_1} = 0.$ A $\gamma_1^*$ 
satisfying this equation is the revenue maximizing arrival
rate to server~1. The admission price corresponding to 
this $\gamma_1^*$ can now be obtained using Eq. \eqref{eq:cgamma}.

\item For each of the three distributions, note that we have 
$\EXP{\bm{\beta}} = 4.$ However, the Revenue rate $R_T$ as
a function of $\gamma_1$ is distinct in all the three cases.
This implies that the revenue rate $R_T$ depends 
on the higher  moments of the distribution $F$ and not just on
its mean value.

\item Finally, note that $R_T$ depends on admission price through the 
addition factor of $c_2 \lambda.$ For different values of $c_2,$ 
the corresponding $\gamma_1^*$ does not change. However it is 
easy to see from Eq. \eqref{eq:cgamma} that $c_1(\gamma_1)$
increases linearly in $c_2.$ 

\end{itemize}

\section{Duopoly}
\label{sec:duopoly}

In this section, we shall consider program \ref{prog:rev_duo_new}
for revenue maximization in the duopoly system. Much of the 
analysis in this section follows from that of the previous section.
Let $\gamma_j, j=1,2$ denote the optimization variable and represent
the admission prices at the respective servers as a function of 
the arrival rates. Towards this, we continue with the use of the
notation $c_j(\gamma_j)$
for $j = 1,2.$ Note that while in the monopoly case, the admission
price $c_2$ was considered fixed, in the duopoly of this section, 
it is the strategy for the second server and hence will
not be a constant. The revenue function for 
Server~$j$ is given by
\begin{equation*}
 R_j(c_j(\gamma_j), \gamma_j) = c_j(\gamma_j) \gamma_j
\end{equation*}
where $c_j(\gamma_j)$ represents the admission price at Server $j$
resulting in an equilibrium arrival rate of $\gamma_j.$ As 
noted in the previous section, $c_j(\gamma_j)$ is a function 
of $c_{j^-},$ the admission price at the other server. 
For a fixed strategy $c_2$ at Server~$2,$ from Eq. \eqref{eq:cgamma} the revenue 
function $ R_1(c_1(\gamma_1), \gamma_1)$ can be redefined as 
\begin{equation}
\label{eq:R1_new}
R_1(c_1(\gamma_1), \gamma_1) = \left(g_1(\gamma_1) + c_2\right)\gamma_1.
\end{equation}
It can be argued as in the previous section that for a fixed $c_1$ 
\begin{equation}
\label{eq:R2_new}
R_2(c_2(\gamma_2), \gamma_2) = \left(g_2(\gamma_2) + c_1\right)\gamma_2
\end{equation}
where 
\begin{equation}
\label{eq:ggamma2}
g_2(\gamma_2) = \beta_1(\lambda - \gamma_2)\left(D_1(\lambda - \gamma_2) - D_2(\gamma_2)\right)
\end{equation}
where from Eq. \eqref{eq:betagamma} $\beta_1(\lambda - \gamma_2)$ 
is as follows
\begin{eqnarray}
\label{eq:betagamma2}
    \beta_1(\lambda - \gamma_2) & = & \begin{cases} 
       F^{-1}\left(\frac{\gamma_2}{\lambda}\right) & 
       \mbox{ for } \lambda - \gamma^+ \leq \gamma_2 \leq \lambda,\\
      F^{-1} \left(\frac{\lambda - \gamma_2}{\lambda}\right) & \mbox{ for } 
        0 < \gamma_2  <  \lambda - \gamma^+.
    \end{cases} 
\end{eqnarray}
It is easy to see that $g_2(\gamma_2)$ is also continuous 
and strictly decreasing in $\gamma_2.$ Further, $g_2(\gamma_2) = 0$ when
$\gamma_2 = \lambda - \gamma^+.$ The revenue maximization
problem for the duopoly is re-stated as follows.

\begin{equation*}
     \label{prog:rev_duo_new2}
   \tag{P7}  
 \begin{aligned}
      & \underset{\gamma_j}{\max}
      & &R_j(\gamma_j) =  \left(g_j(\gamma_j) + c_{j^-}\right)\gamma_j \\
      & \text{subject to} & &  0 \leq \gamma_j \leq \gamma^j(c_{j^-}) \leq \lambda \\
      & \text{given} & & c_{j^-}.
    \end{aligned}
  \end{equation*}  

 For a given $c_{j^-},$ recall that $\gamma_j^*(c_{j^-})$ denotes the maximizer of program 
 \ref{prog:rev_duo_new} and hence of the above program.  Also recall that 
 $\hat{B}_j(c_{j^-})$ denotes the best response admission price at 
 Server~$j$ in response to the admission price $c_{j^-}$ at the other facility.
 Then the Nash equilibrium set of admission prices, denoted by $(c^*_1, c^*_2),$ 
 is characterized as follows.
 
 \begin{eqnarray}
 \label{eq:nash}
  (c_1^*, c_2^*) = \left\lbrace (c_1, c_2) : \hat{B}_1(c_2) = c_1, \hat{B}_2(c_1) = c_2\right\rbrace,
 \end{eqnarray} 
  where $\hat{B}_j(c_{j^-}) = g_j(\gamma_j^*(c_{j^-})) + c_{j^-}$ for $j=1,2.$
 
 We begin the analysis for the duopoly problem by first identifying that 
 $\gamma_j^*(c_{j^-})$ lies in the interior of the domain. We have the following lemma.
\begin{lemma}
\label{lem:gamma_1^*}
 $\gamma_j^*(c_{j^-}) \notin \left\lbrace 0, \gamma^j(c_{j^-}) \right\rbrace.$
\end{lemma}
 \begin{proof}
 See Appendix for proof.
\end{proof}

 For a given $c_{j^-}$ since  $\gamma_j^*(c_{j^-})$ lies in the interior
 of the domain,  $\gamma_j^*(c_{j^-})$ satisfies $\frac{dR_j}{d\gamma_j}\bigg\vert_{\gamma_j = \gamma_j^*}  = 0$
 and $\frac{d^2R_j}{d\gamma^2_j}\bigg\vert_{\gamma_j = \gamma_j^*} \leq 0.$ 
 
 Define $S_j(c_{j^-}):= \left\lbrace \gamma_j : \frac{dR_j}{d\gamma_j}  = 0, 
 \frac{d^2R_j}{d\gamma^2_j} \leq 0\right\rbrace$. Then $\gamma_j^*(c_{j^-})$
 is obtained as a solution to the following.

 \begin{equation*}
     \label{prog:gamma^*}
   \tag{P8}  
 \begin{aligned}
      \gamma_j^*(c_{j^-})  = & \underset{\gamma_j \in S_j(c_{j^-})}{\arg\max}
      & &R_j(\gamma_j). \\
      \end{aligned}
  \end{equation*}  

 From the above discussion, it should be clear that obtaining the closed form expression for 
 $(c_1^*, c_2^*)$ satisfying the simultaneous equations of \eqref{eq:nash} is, in general,
 not easy. Note that our analysis till now makes minimal assumptions on the distribution function
 $F$ or on the delay function $D_j(\cdot).$ For certain choices of these functions,
 it may be difficult to obtain a closed form expression for $\gamma_j^*(c_{j^-}).$ 
 The objective function $R_j$ also need not be a concave function. In that case, a brute 
 force search among all the local maxima points needs to be carried out to choose the 
 right $\gamma_j^*(c_{j^-}).$ Instead of satisfying ourselves with some numerical examples,
in the following subsection we shall analyze the Nash equilibrium under the restriction 
that the two servers are identical i.e.,  the average delay at any queue is the same for 
the same arrival rate. Under this setting, our interest is to characterize the symmetric
Nash equilibrium such that $c_1^* = c_2^*.$ 

\subsection{Characterizing a symmetric Nash equilibrium} 
In this section we shall characterize the necessary conditions 
for the existence of a symmetric Nash equilibrium, i.e., 
$(c_1^*, c_2^*)$ where $c_1^* = c_2^*:= c^*.$ A natural scenario
where such an equilibrium is possible is when the two servers have
identical delay functions. In this section, we restrict to this 
case and assume that $D_j(\cdot) = D(\cdot)$ for $j = 1,2.$
As the service systems are identical in their delay characteristics,
it is desirable to identify conditions for existence of a 
symmetric Nash equilibrium. We begin with the following definition.
Define $\alpha_1, \alpha_2$ as follows.
\begin{eqnarray}
 \label{eq:alpha} 
 \alpha_1 &=& -\gamma^+ \frac{dg_1(\gamma_1)}{\gamma_1}\bigg\vert_{\gamma_1 = \gamma^+}  \nonumber \\
 \alpha_2 &=& -\gamma^+ \frac{dg_2(\gamma_2)}{\gamma_2}\bigg\vert_{\gamma_2 = \gamma^+} \nonumber \\
 \end{eqnarray}
Based on these definitions, we have the following lemma.

\begin{lemma}
\label{lem:alpha}
 $\alpha_1 = \alpha_2.$
\end{lemma}
 \begin{proof}
From the definition 
 of $g_1(\cdot)$ in Eq. \eqref{eq:ggamma} we have 
 \begin{eqnarray*}
  \frac{dg_1(\gamma_1)}{\gamma_1} &=& \beta'_1(\gamma_1) \left(D_2(\lambda - \gamma_1) - D_1(\gamma_1) \right) \\
  &+&  \beta_1(\gamma_1) \left( D'_2(\lambda - \gamma_1) - D'_1(\gamma_1) \right)
 \end{eqnarray*}
where the partial derivatives on the r.h.s. are w.r.t. $\gamma_1.$ Now from the definition of 
$\gamma^+$ and from Eq. \eqref{eq:betagamma} we have 
\begin{equation*}
 \frac{dg_1(\gamma_1)}{\gamma_1}\bigg\vert_{\gamma_1 = \gamma^+} =  F^{-1}\left(\frac{\lambda - \gamma^+}{\lambda}\right)
 \left( D'_2(\lambda - \gamma^+) - D'_1(\gamma^+) \right).
\end{equation*}
Similarly, from the definition of $g_2(\cdot)$ we have 
 \begin{eqnarray*}
  \frac{dg_2(\gamma_2)}{\gamma_2} &=& \beta'_1(\lambda - \gamma_2) \left(D_1(\lambda - \gamma_1) - D_2(\gamma_2) \right) \\
  &+&  \beta_1(\lambda - \gamma_2) \left( D'_1(\lambda - \gamma_2) - D'_2(\gamma_2) \right)
 \end{eqnarray*}
where the partial derivatives on the r.h.s are now w.r.t $\gamma_2.$
Note that since $\gamma_1 = \lambda - \gamma_2,$ we have 
$\frac{\partial D_1(\gamma_1)}{\partial \gamma_1} = -\frac{\partial D_1(\gamma_1)}{\partial \gamma_2}.$
Further note that since the servers are identical, i.e.,
$D_j(\cdot) = D(\cdot)$ for $j = 1,2$ from the definition of 
$\gamma^+$ we have $\gamma^+ = \frac{\lambda}{2}.$ 
From Eq. \eqref{eq:betagamma2} and the fact that $\lambda - \gamma^+ = \gamma^+$,
we have  
\begin{equation*}
 \frac{dg_2(\gamma_2)}{\gamma_2}\bigg\vert_{\gamma_2 = \gamma^+} =  F^{-1}\left(\frac{\lambda - \gamma^+}{\lambda}\right)
 \left( D'_2(\lambda - \gamma^+) - D'_1(\gamma^+) \right).
\end{equation*}
This proves that $\alpha_1 = \alpha_2.$
\end{proof}
We now have the following theorem, that
characterizes the necessary condition for a symmetric Nash 
equilibrium.

\begin{theorem}
\label{thm:symmetric_nash}
Let $c_1^* = c_2^*$ be a symmetric Nash equilibrium for the duopoly 
price competition. Then $c_1^* = c_2^* = \alpha_1$.
\end{theorem}
\begin{proof}
Recall that the Nash equilibrium is characterized as 
 \begin{eqnarray*}
  (c_1^*, c_2^*) = \left\lbrace (c_1, c_2) : \hat{B}_1(c_2) = c_1, \hat{B}_2(c_1) = c_2\right\rbrace.
 \end{eqnarray*} 
where $\hat{B}_j(c_{j^-}) = g_j(\gamma_j^*(c_{j^-})) + c_{j^-}$ for $j=1,2.$
This implies that $c_j^* = g_j(\gamma_j^*(c^*_{j^-})) + c^*_{j^-}$
and since $c_1^* = c_2^*,$ we have $g_j(\gamma_j^*(c^*_{j^-})) = 0$
for $j = 1,2.$ Now from Lemma \ref{lemma:gfx} and symmetry of the servers,
this implies that $\gamma_1^*(c^*_{2}) = \gamma^+$
and $\gamma_2^*(c^*_{1}) = \lambda - \gamma^+.$ Since 
the servers are identical, we have $\gamma^+ = \frac{\lambda}{2}$
and hence $\gamma_2^*(c^*_{1}) = \gamma^+.$
Since $\gamma_j^*(c_{j^-})$ is also a solution to program \ref{prog:gamma^*},
 $\gamma_j^*(c_{j^-}) \in S(c_{j^-}).$ From the definition of 
$S(c_{j^-}),$ this implies that 
$\frac{dR_j}{d\gamma_j}\bigg\vert_{\gamma_j = \gamma^+} = 0.$
Further, this implies from the definition of $R_j(\gamma_j)$ that
 \begin{eqnarray}
 \label{eq:brp}
  \frac{dR_j}{d\gamma_j}\bigg\vert_{\gamma_j = \gamma^+} &=& 
 \gamma^+ \frac{dg_j(\gamma_j)}{\gamma_j}\bigg\vert_{\gamma_j =
 \gamma^+} +g_j(\gamma^+)+ c^*_{j^-} \\
  & = & 0. \nonumber
 \end{eqnarray}
We have $g_j(\gamma^+) = 0$ and hence from the definition of $\alpha_j$ 
for $j  =1,2 $ we have $c_{j^-}^* = \alpha_j.$ From Lemma
\ref{lem:alpha}, we have $\alpha_1 = \alpha_2$ and hence  
$c_1^* = c_2^* = \alpha_1.$ This completes the proof. 
\end{proof}

Note that the above theorem only provides a necessary condition
for the Nash equilibrium pair and we shall soon see that in fact 
this condition is not sufficient. We shall now provide a few 
examples illustrating the occurrence of symmetric Nash equilibria.

\begin{figure}
  \begin{center}
    \includegraphics[height=2in]{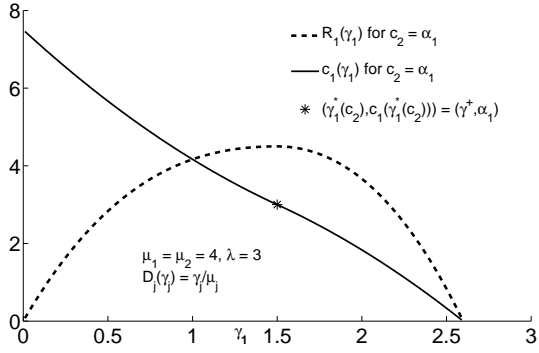}
  \end{center}
  \caption{$R_1$ and $c_1(\gamma_1)$ when $D_j(\gamma_j) =
  \frac{\gamma_j}{\mu_j}$ and $F(\cdot)$ is Uniform over $[2,6]$ }
  \label{fig:ex3}
\end{figure}

\textbf{Example 3:} In this example, we assume that 
$D_j(\gamma_j) = \frac{\gamma_j}{\mu_j}$ for $j = 1,2.$ 
Let $\mu_1 = \mu_2 = 4$ while the arrival rate is $\lambda = 3.$ 
We suppose that the distribution $F(\cdot)$ has a uniform
distribution with support of $[a,b].$ We plot $R_1(\gamma_1)$ and
$c_1(\gamma_1)$ as a function of $\gamma_1$ in Fig.\ref{fig:ex3}.
The aim of this example is to check whether 
$(c_1^*, c_2^*) = (\alpha_1, \alpha_1)$ is a symmetric
Nash equilibrium. For the set of parameters of this example 
we have $\gamma^+ = 1.5$ and since  
$$
\alpha_1 = -\gamma^+ \frac{dg_1(\gamma_1)}{\gamma_1}\bigg\vert_{\gamma_1 = \gamma^+}
$$
we have $\alpha_1 = 3.$ We now set $c_2 = \alpha_1 = 3.$ 
Clearly, for a symmetric Nash equilibrium 
$(c_1^*, c_2^*) = (\alpha_1, \alpha_1)$, $\gamma_1^*(c_2) = \gamma^+ = 1.5$ 
must hold. It is easy to see from Fig. \ref{fig:ex3} that 
$R_1(\gamma_1)$ is indeed maximized when $\gamma_1 = \gamma^+$
implying that $\gamma_1^*(c_2) = \gamma^+.$ Further it can be 
verified that $(c_1(\gamma_1^*(c_2))) \alpha_1.$ Clearly, 
$(c_1^*, c_2^*) = (\alpha_1, \alpha_1)$ is a symmetric
Nash equilibrium for this example.

\textbf{Example 4:} With the help of this example, we will
illustrate that the necessary conditions stated in the previous 
theorem need not be sufficient. We shall once again assume that 
$D_j(\gamma_j) = \frac{\gamma_j}{\mu_j}$ where  
$\mu_1 = \mu_2 = 4.$ As for the choice of $F(\cdot),$ we 
consider an exponential distribution with $\tau = 4.$ 
A plot of $R_1(\gamma_1)$ and $c_1(\gamma_1)$ as a function 
of $\gamma_1$ is provided in Fig. \ref{fig:ex4}. For this 
example we start by setting $c_2 = \alpha_1.$ However we 
observe that the best response $\gamma_1^*(c_2) \neq \gamma^+$
and hence $c_1(\gamma_1^*(c_2)) \neq c_2.$ Both these points 
$\gamma_1^*(c_2),c_1(\gamma_1^*(c_2)$ and $(\gamma^+, \alpha_1)$
are represented in Fig. \ref{fig:ex4}. Clearly, 
$(\alpha_1, \alpha_1)\neq (c_1^*, c_2^*)$ and therefore 
the sufficiency conditions differ from the necessary ones.

\begin{figure}
  \begin{center}
    \includegraphics[height=2in]{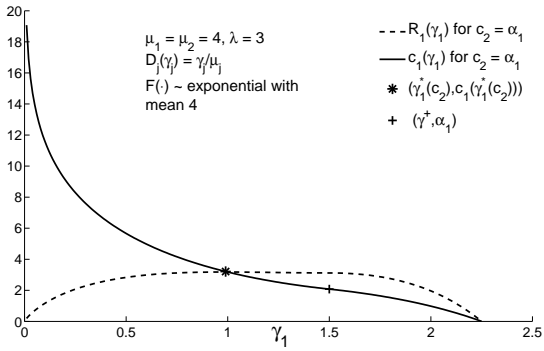}
  \end{center}
  \caption{$R_1$ and $c_1(\gamma_1)$ when $D_j(\gamma_j)
  = \frac{1}{\mu_j - \gamma_j}$ and $F(\cdot)$ is exponential with $\tau=4$}
  \label{fig:ex4}
\end{figure}


%
%
%
%

\section{Estimating the distribution $F$}
\label{sec:estimate_F}
Recall that $F$ denotes the distribution function for the delay
sensitivity of the arriving customers. The knowledge of $F$ is
necessary to determine the equilibrium kernel $K^W$ introduced
in Theorem \ref{thm:wardrop-cts}. Further, the kernel $K^W$ must
be known for the revenue maximization problems seen in 
this paper. However in most practical situations, the 
distribution function $F$ may not be known and due to the 
unobservable nature of the queues it may not be possible to even
elicit such information from the arriving systems.   
In such situations the only alternative may be to estimate this 
distribution function. One possible method to do so is to vary 
the admission prices at the servers and then measure the change
in the arrival rate of customers to the different server and then
use the Wardrop equilibrium conditions to estimate $F$.
In this section, we shall describe a simple procedure to 
estimate the underlying continuous distribution function $F.$ 
Our proposed method is well suited for a monopoly system
when the single service provider has access to both the 
admission prices. In this section, we also consider the case
when $\bm{\beta}$ is a discrete random variable. In this case,
the customers are divided into finite number of classes 
differing in their values of $\bm{\beta}.$ The aim is to identify
the value of $\bm{\beta}$ for the different classes along with the
Poisson arrival rates $\lambda_i$ for the classes. 
Refer \cite{Bodas14,Armony03,Mandjes07} for some examples of 
service systems where such discrete customer classes are 
considered.

Throughout this section, we shall make the following assumptions.
We shall assume that the two servers are modeled as $M/M/1$ queues
with service rates $\mu_1$ and $\mu_2$ and admission prices $c_1$ 
and $c_2$ respectively. With this assumption, we have 
$D_j(\gamma_j) = \frac{1}{\mu_j - \gamma_j}.$ It goes without saying
that our analysis will also hold for any delay cost $D_j(\cdot)$ 
that is monotonic and strictly increasing 
in its arguments. We assume that once the admission prices $c_1$
and $ c_2$ at the servers are announced and that the Wardrop equilibrium is
achieved, each server $j$ will accurately determine
or measure the equilibrium arrival rate $\gamma_j$ and the mean delay cost
$D_j(\gamma_j)$ for $j = 1,2.$ Hence the measured values $\gamma_j$ 
and $D_j(\gamma_j)$ and the the corresponding quantities at the Wardrop 
equilibrium will be assumed to be the same. We also assume that the total arrival
rate of customers to the system denoted by $\lambda$ is known a priori 
and that $c_1 > c_2,$ i.e., the admission price at
the first server is higher than the second. Note that since the distribution
$F(\cdot)$ is unknown, the functions $ \beta_1(\cdot), g_1(\cdot), c_1(\cdot)$
also cannot be determined and used for our procedure.


We begin by estimating the distributions $F$ that belongs to a
parameterized family, say for example the exponential distribution.
Let the parameter for the exponential distribution be denoted by $\alpha.$
When $c_1$ and $c_2$ at the two servers are fixed, the equilibrium 
$\gamma_1$ and $\gamma_2$ at the servers is measured immediately. 
We choose a $c_1, c_2$ such that $\gamma_j > 0$ for $j=1,2.$
From this, the mean delay cost $D_j(\gamma_j)$ for $j = 1,2$ is also 
calculated. Since all the quantities (except $\beta_1$) in
Eq. \eqref{eq:beta_1} of Theorem \eqref{thm:wardrop-cts} are known, 
the threshold $\beta_1$ can be determined as 
$\beta_1 = \frac{c_1 - c_2}{D_2(\gamma_2) - D_1(\gamma_1)}.$
Now increase $c_1$ to $c_1^1$ where $c_1^1 = c_1 + \delta$ for $\delta > 0.$
This decreases the equilibrium $\gamma_1$ to say $\gamma_1^\delta.$
Let $\beta_1^\delta$ denote the threshold when the arrival rate to Server~1
is $\gamma_1^\delta.$ Again, using the measurements of the arrival rates 
and the delay functions $\beta_1^\delta$ can be determined from 
Eq. \eqref{eq:beta_1}. Since $\gamma_1^\delta <  \gamma_1 < \gamma^+$,
from Lemma \ref{lemma:betagamma}, we know that $\beta_1(\gamma_1^\delta) 
> \beta_1(\gamma_1).$ This implies that $\beta_1^\delta > \beta_1.$
Clearly, the ratio $\frac{\gamma_1
  - \gamma_1^\delta}{\lambda}$ is the probability of an arriving
customer with $\beta \in [\beta_1, \beta_1^\delta]$ and hence

\begin{equation}
\label{eq:exp_para}
  \int_{\beta_1}^{\beta_1^\delta}\alpha e^{-x\alpha}dx =
  \frac{\gamma_1 - \gamma_1^\delta}{\lambda}.
\end{equation}
The only unknown quantity is the exponential parameter $\alpha$ which 
can now be obtained from the above equation. 

\begin{remark}
Since the exponential distribution has a single parameter, the 
parameter could be obtained using only Eq. \eqref{eq:exp_para}. 
For a parameterized distribution with $k$ parameters, we need 
$k$ simultaneous equations in terms of the underlying parameters.
These can be obtained by following the procedure above for $k$ different 
admission price $\left\lbrace c_1^k \right\rbrace$ at Server~1. 
\end{remark}

We will now describe a numerical method to obtain a piecewise constant 
approximation for the density function $f$ that is not necessarily from 
a parameterized family of distribution functions.  As an example, consider a random 
variable $\bm{\beta}$ supported on the range $[0,4].$ Suppose the 
distribution function is
$$
P(\bm{\beta} \leq x) =  F(x) = \frac{x^2}{16}.  
$$
The corresponding density function is denoted by $f(x)$ is $x/8$
for $x \in [0,4].$ For this example assume that there are two $M/M/1$
servers with service rates $\mu_1 = 5$ and $\mu_2 = 5,$ admission
prices initially set to $c_1 = c_2 = 5$ and the total arrival rate
$\lambda = 5$. As earlier, we assume that once the admission prices
at the servers are announced, the Wardrop equilibrium is reached
instantaneously and each servers can accurately determine the
aggregate arrival rates and the mean delay per customer.

Increase $c_1$ by $\delta > 0$ and for the admission price vector
$(c_1 + \delta,c_2),$ measure the equilibrium arrival rates and the mean delay in
the queues and calculate the corresponding threshold $\beta_1$ using
Eq. \eqref{eq:beta_1}. Repeat this for a finite number of
times, each time increasing $c_1$ from its previous value by $\delta.$
This experiment is denoted in Table \ref{table:measure}.

\begin{table}[h]
\centering
\begin{tabular}{|l|l|l|l|}
\hline
$c_1$ & $c_2$ & $\gamma_1$ & $\beta_1$ \\ \hline
5.0   & 5     & 1.98       & 2.84      \\ \hline
5.2   & 5     & 1.69       & 3.04      \\ \hline
5.4   & 5     & 1.44       & 3.20       \\ \hline
5.6   & 5     & 1.23       & 3.33      \\ \hline
5.8   & 5     & 1.05       & 3.44      \\ \hline
6.0   & 5     & 0.89       & 3.53      \\ \hline
6.2   & 5     & 0.75       & 3.60       \\ \hline
6.4   & 5     & 0.63       & 3.67      \\ \hline
6.6   & 5     & 0.52       & 3.37      \\ \hline
6.8   & 5     & 0.43       & 3.78      \\ \hline
\end{tabular}
\caption{The table indicates the price vector $(c_1, c_2)$, the measured value of 
  $\gamma_1$ and the threshold $\beta$ obtained from Eq. \eqref{eq:beta_1}.}
\label{table:measure}
\end{table}

Using the earlier notation, we observe from the table that 
as $c_1$ increases to, say $c_1 + \delta$, 
$\gamma_1$ decreases to $\gamma_1^{\delta}$ while 
the threshold $\beta_1$ increases (to $\beta_1^{\delta}$).
As earlier, we have 
\begin{equation*}
 \int_{\beta_1}^{\beta_1^\delta}f(x)dx = \frac{\gamma_1 - \gamma_1^\delta}{\lambda}
 \end{equation*}
where the density function $f(x)$ is to be estimated.
Assume for all $x \in ({\beta_1},{\beta_1^\delta})$ that 
$f(x)= z,$ where $z$ is a constant. By assuming this, we are approximating 
the density function $f(x)$ for $x \in ({\beta_1},{\beta_1^\delta})$
by a horizontal line of magnitude $z$ and thus approximating 
$f(x)$ by a piecewise constant function. As $\delta \rightarrow 0,$ the 
approximation should converge to the true density function.
We now have 
\begin{equation}
z = \frac{\gamma_1 - \gamma_1^\delta}{\lambda({\beta_1^\delta} - \beta_1).}
\label{eq:z} 
\end{equation}
The value of $z$ for a fixed $c_1$ and $c_1 + \delta$ can be viewed as an estimate
for the density function $f(x)$ and obviously $z \rightarrow f(x)$ as
$\delta \rightarrow 0.$ These values of $z$ for different values of
$c_1$ are given in Table \ref{table:zvalue}.
\begin{table}[h]
\centering
\begin{tabular}{|l|l|l|}
\hline
$c_1$ & $c_1 + \delta$ & $z$  \\ \hline
5     & 5.2            & 0.37 \\ \hline
5.2   & 5.4            & 0.39 \\ \hline
5.4   & 5.6            & 0.41 \\ \hline
5.6   & 5.8            & 0.42 \\ \hline
5.8   & 6.0            & 0.44 \\ \hline
6.0   & 6.2            & 0.44 \\ \hline
6.2   & 6.4            & 0.45 \\ \hline
6.4   & 6.6            & 0.46 \\ \hline
6.6   & 6.8            & 0.47 \\ \hline
\end{tabular}
\caption{The estimates $z$ can be obtained from Eq.~\eqref{eq:z} from
the successive changes in the admission prices and the corresponding 
measurements of the arrival rates.}
\label{table:zvalue}
\end{table}

\begin{figure}
  \begin{center}
    \includegraphics[height=2in]{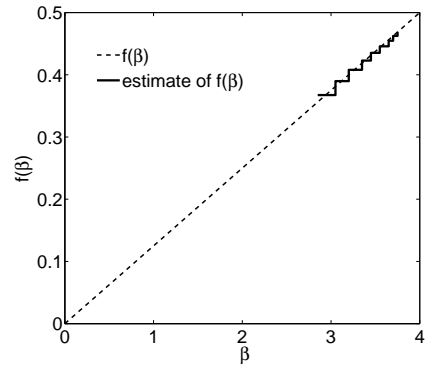}
  \end{center}
  \caption{Comparing the estimate of $f(\cdot)$ with the true density function.}
  \label{fig:feb9_mainfig}
\end{figure}
A plot comparing the true density function and the estimate is given
in Fig. \ref{fig:feb9_mainfig}. The plot shows that the estimate of the density
function is reasonably accurate and for better estimation, one
naturally required more of such measurement points.

There is however a limitation to this method. Note that 
when $c_1 = c_2,$ the corresponding value of $\beta_1 = 2.84.$ Any
increase or decrease in either $c_1$ or $c_2$ cannot result in a
$\beta_1$ such that $\beta_1 < 2.84.$ This is because, for the 
underlying distribution we have from Eq. \eqref{eq:betagamma} that 
$\beta_1 (\gamma^+) = 2.84$ and for any $\gamma \in [0, \lambda]$
with $\gamma \neq \gamma^+,$ we have $\beta_1 (\gamma) > \beta_1(\gamma^+).$
As a result, the density function $f(x)$ cannot be estimated for 
$x \leq 2.84.$ 

\subsection{Estimating Discrete Distribution}

\begin{figure}
    \begin{center}
      \includegraphics[width=.99\columnwidth]{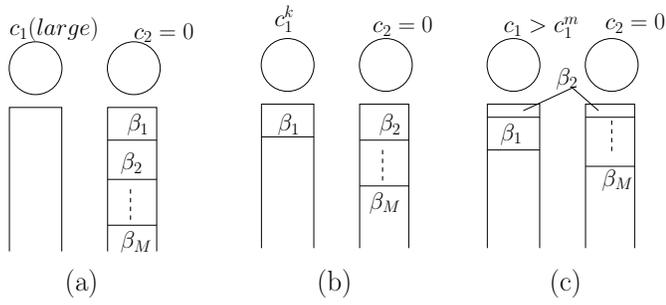}
    \end{center}
    \caption{Estimating the discrete distribution $F$}
    \label{fig:estimate_c}
\end{figure}

We shall now consider the case where
the distribution $F$ is a discrete distribution with $M$ point masses.
Thus, there are $M$ customer classes and we will assume that 
for each Class $i,$ the associated waiting cost $\beta_i$ and the
arrival rate $\lambda_i$ are unknown. Further, $\beta_1 > \beta_2 \ldots > \beta_M.$
See \cite{Bodas14} for the analysis of Wardrop equilibrium of such a model.
We continue with the assumption that there are two servers each charging an
admission price $c_1$ and $c_2.$ We begin by setting $c_2 = 0$ and 
$c_1$ to an arbitrarily large value such that $\gamma_1 = 0$ while
$\gamma_2 = \lambda.$ This is represented in part (a) of Fig. 
\ref{fig:estimate_c}. It goes without saying that the necessary assumption 
is that $\mu_2 > \lambda.$  Now start decreasing $c_1$ in 
steps of size $\delta$ and stop at the first instance when 
$\gamma_1$ increases to an arbitrarily small value $\epsilon.$
We use the notation $c^j_1$ and $\gamma^j_1$ to denote the admission 
price and the arrival rate at Server~$1$ when $c_1$ is decreased $j$ 
times by $\delta$, i.e., when $c^j_1 = c_1 - j\delta.$ 
$\gamma_1 = \epsilon$ implies that the most sensitive delay class $\beta_1$
must now be using Server~1 along with Server~2. 
Since the delay function at each queue can be measured,
$\beta_1$ can be easily determined from the corresponding 
Wardrop condition  
$$
c^j_1 + \beta_1 D_1(\gamma^j_1) = \beta_1 D_2(\gamma^j_2).
$$
 We will now determine $\lambda_1$
corresponding to this $\beta_1.$ Continue decreasing $c_1.$   
The proportion of Class~$1$ customers using Server~$1$ keeps 
increasing till all Class~$1$ customers use only
Server~$1.$ When this happens, the corresponding Wardrop 
equilibrium condition for some $k > j$ satisfies 
$$
c^k_1  <  \beta_1 \left(D_2(\gamma^k_2)-D_1(\gamma^k_1)\right)
$$
and this is represented by part (b) in Fig. \ref{fig:estimate_c}.
For a Class 2 customer to start using Server~1, the Wardrop equilibrium
condition is
$$
c^m_1  = \beta_2 \left(D_2(\gamma^m_2)-D_1(\gamma^m_1)\right)
$$
where $m> k.$ Further since $m> k,$ we have
 $$
\beta_2 \left(D_2(\gamma^m_2)-D_1(\gamma^m_1)\right) < \beta_1
\left(D_2(\gamma^k_2)-D_1(\gamma^k_1)\right).
$$
and hence for all $l$ such that $k < l < m,$ we have
 $$
\beta_2 \left(D_2(\gamma^m_2)-D_1(\gamma^m_1)\right) < c^l_1< 
\beta_1 \left(D_2(\gamma^k_2)-D_1(\gamma^k_1)\right).
$$
This means that for any $c^l_1$ satisfying $c^k_1 < c^l_1< c^m_1,$ 
$\gamma^l_1$ and $\gamma^l_2$ remain unchanged. 
Clearly in this case $\lambda_1 = \gamma^l_1.$ Fig. \ref{fig:estimate_c},
part (c) represents the fact that
for any $c_1 > c_1^m,$ Class~2 customers use both the servers at Wardrop 
equilibrium. Continue this process
till all the $\lambda_i,$ $\beta_i$ as well as
the number of customer classes $M$ is determined.
It should be noted that the accuracy of our method increases
as $\delta \rightarrow 0.$ A downside of a small $\delta$
is that the procedure may take a very long time 
to discover the system parameters.

\section{Summary and Future work}

In this paper, we have considered the problem of revenue maximization 
in parallel server systems. We specialize with the case of two servers 
and first assume the case when both the servers belong to the same 
service provider. The admission price at one of the server is required 
to be fixed and the service system can change the admission price at 
the other server to maximize its revenue. The Wardrop equilibrium when
customers are heterogeneous and strategic has already been characterized
in our earlier paper. We use this characterization to simplify the revenue 
maximization program to make it more amenable to analysis. The equivalent
program is easy to interpret, analyze and provides more insight into 
the problem. While it is intuitive that for a fixed $c_2,$ the revenue  
maximizing $c_1$ should always be greater than $c_2,$ the program enables 
to characterize the revenue maximizing $c_1^*$ as a function of $c_2.$

In the second part of the paper, we consider the duopoly model where 
each server competes with the other one to maximize its revenue. This is 
a standard game-theoretic problem and the aim is to identify the Nash 
equilibrium set of prices. We see however that since the customers are 
heterogeneous, the first order necessary conditions are not easy to 
solve. Instead, we characterize this Nash equilibrium for a simplified 
case when the two servers are identical in their delay characteristics. 
In this case we are interested in the symmetric Nash equilibrium prices. 
We provide the necessary condition for this case and identify the Nash
equilibrium prices for different distributions $F$ and delay cost 
functions $D(\cdot).$ 

In both these problems problems and also in the social welfare maximization
problem of our previous paper, an important assumption is that
the distribution function $F$ is known. We relax this assumption in Section 
\ref{sec:estimate_F} and provide a procedure to estimate this 
distribution. The proposed method is of course preliminary and assumes
that one is allowed to change admission price any number of time to 
measure the change in the equilibrium arrival rate. Further, we have 
assumed that there is no cost to making such measurements. A more realistic
method incorporating these practical limitations may make the problem
more relevant and this is part of future work.

\section*{Appendix}

\textbf{Lemma \ref{lemma:compare_gamma_c}} \\
\begin{proof}
We first prove that $\gamma_1 \in [0,\gamma^+]$ implies
$c_1 \geq c_2.$ Recall the definition
of $\gamma^+$ that $$\gamma^+ = \left\lbrace \gamma_1 : D_1(\gamma_1) =
D_2(\gamma_2)\right\rbrace.$$
Since $D_j(\gamma_j)$ is monotonic and increasing in $\gamma_j$ for 
$j = 1,2 $ and that $\gamma_2 = \lambda - \gamma_1$
we have $D_1(\gamma_1) \leq D_2(\gamma_2)$ for $\gamma_1 \in 
[0,\gamma^+].$ Now let $\gamma_1 = 0.$ Since no customer uses Server~1
at equilibrium, this implies that $c_1 + \beta D_1(0) > c_2 + 
\beta D_2(\lambda)$ for all $\beta.$ Since $D_1(0) < D_2(\lambda)$ (assumption)
$c_1 > c_2$ must be true.

When $\gamma_1 = \gamma^+,$ we will show that $c_1 = c_2.$ 
Suppose this is not true, i.e.,  $\gamma_1 = \gamma^+$
while $c_1 \neq c_2$. $\gamma_1 = \gamma^+$ implies
$D_1(\gamma_1) = D_2(\gamma_2).$ As $c_1 \neq c_2,$ customers have an
incentive to move from the server with a higher admission price
to the one with a lower price. This implies that $\gamma_1 = 
\gamma^+$ is not an equilibrium and this is a contradiction.

Now consider $\gamma_1 \in (0,\gamma^+)$ where $D_1(\gamma_1) < 
D_2(\gamma_2).$ From Theorem~\ref{thm:wardrop-cts}, 
$\gamma_1 \in (0,\gamma^+)$ implies $\beta_1 \in (a,b)$
and hence $c_1 + \beta_1 D_1(\gamma_1) = c_2 + \beta_1 D_2(\gamma_2).$
Since $D_1(\gamma_1) < D_2(\gamma_2)$ we have $c_1 \geq c_2.$

We now prove that if $c_1 \geq c_2,$ then $\gamma_1 \in [0,\gamma^+].$
We first show that when $c_1 = c_2,$ we have $\gamma_1 = \gamma^+.$
Suppose that when $c_1 = c_2,$ $\gamma_1 \neq \gamma^+.$ From the 
definition of $\gamma^+$ we have $D_1(\gamma_1) \neq D_2(\gamma_2)$
and hence customers have an incentive to move from the server with higher
expected delay to the one with lower expected delay. This implies that
when $c_1 = c_2,$ $\gamma_1 \neq \gamma^+$ is not an equilibrium.

Now let $c_1 > c_2.$ From Theorem \ref{thm:wardrop-cts} we have either
$\beta_1 = a$ or $\beta_1 = b$ or $\beta_1 \in(a,b).$ 
The case $\beta_1 = a$ corresponds to the case when all customers choose 
Server~2 at equilibrium and this cannot happen! This is because
while $c_1 > c_2,$ we have also assumed $D_1(\lambda) > D_2(0).$ 
$K^W$ with $\beta_1 = a$ will be possible only if
$$c_1 - c_2 \leq \beta (D_2(0) - D_1(\lambda))$$ for all $\beta \in [a,b].$
Now this is not possible as the left hand side is positive while 
the right hand side is negative.
It is straightforward to see that when $\beta_1 = b,$ we have 
$\gamma_1 = 0 $ and hence $\gamma_1 \in [0, \gamma^+].$
When $\beta_1 \in(a,b)$ we have $c_1 + \beta_1 D_1(\gamma_1) =
c_2 + \beta_1 D_2(\gamma_2).$ Again, since $c_1 > c_2,$ we have
$D_1(\gamma_1) \leq D_2(\gamma_2)$ and this requires $\gamma_1 \in (0,\gamma^+).$
The proof for $ \gamma_1 \in (\gamma^+,\lambda]$ follows along similar
lines and will not be provided. This completes the proof.
\end{proof}

\textbf{Lemma \ref{lemma:betagamma}} \\
\begin{proof}
From Lemma \ref{lemma:compare_gamma_c}, $\gamma_1 \in [0,\gamma^+)$
implies that $c_1 > c_2$ while  $\gamma_1 \in (\gamma^+,\lambda]$
implies $c_1 < c_2.$ Now from Theorem \ref{thm:wardrop-cts}, when $c_1 > c_2,$
we have  
$$
\gamma_1 = \lambda \int_{\beta_1}^{b} 1 dF(\beta) = \lambda(1 - F(\beta_1)).
$$
Similarly, when $c_1 < c_2$ we have 
$$
\gamma_1 = \lambda \int_{0}^{\beta_1} 1 dF(\beta) = \lambda(F(\beta_1)).
$$
Now $\beta_1(\gamma_1)$ defined as the value of threshold $\beta_1$ 
when the equilibrium arrival rate to Server~1 is $\gamma_1$ can be 
represented as follows.
\begin{eqnarray}
 \beta_1(\gamma_1) & = & \begin{cases} 
 \beta : \int_{\beta}^b\lambda dF(\beta) = \gamma_1 & \mbox{ for } 0\leq \gamma_1 < \gamma^+,\\
 \beta : \int_a^{\beta}\lambda dF(\beta) = \gamma_1 & \mbox{ for } \gamma^+ < \gamma_1 < \lambda.
 \end{cases} 
\end{eqnarray}
Now as seen earlier, $F(\cdot)$ is absolutely 
continuous and strictly increasing in its domain. Further,
the support is $[a,b]$ and hence $F(\cdot)$ is a bijective
function whose inverse exists. In fact
$F^{-1}(\cdot)$ is continuous and strictly increasing in its domain.
The statement of the lemma now follows.
 \end{proof}

\textbf{Lemma \ref{lemma:gfx}} \\
 \begin{proof} Recall our assumption that $D_j(\gamma_j)$ is continuous and
 monotone increasing in
$\gamma_j$ where $j=1,2.$ Since $\gamma_2 = \lambda - \gamma_1,$ 
$(D_2(\lambda - \gamma_1) - D_1(\gamma_1))$  is monotone decreasing in $\gamma_1$ for 
$0 \leq \gamma_1 \leq \lambda.$
Recall Eq. \eqref{eq:betagamma} that determines $\beta_1(\gamma_1).$ 
For $0 \leq \gamma_1 < \gamma^+,$ $\beta_1(\gamma_1)$ is continuous
and strictly decreasing. The continuity follows from that of $F^{-1}(\cdot).$ 
Since $F^{-1}(\cdot)$ is strictly increasing in its arguments, 
$F^{-1}\left(\frac{\lambda - \gamma_1}{\lambda}\right) = \beta_1(\gamma_1)$
is decreasing in $\gamma_1.$ Clearly, $g_1(\gamma_1)$ is monotone decreasing when
$\gamma_1$ is such that $0 \leq \gamma_1 < \gamma^+.$ 

When $\gamma_1$ is such that $\gamma^+ <  \gamma_1 \leq \lambda,$ from the 
definition of $\gamma^+,$ we have $(D_2(\lambda - \gamma_1) - D_1(\gamma_1)) < 0.$
In this range of $\gamma_1,$ it can be seen from Eq. \eqref{eq:betagamma} that 
$\beta_1(\gamma_1)$ is continuous and increasing in $\gamma_1$. This again implies that 
$g_1(\gamma_1)$ is continuous decreasing when $\gamma_1$ satisfies 
$\gamma^+ <  \gamma_1 \leq \lambda.$

$g_1(\gamma^+) = 0$ follows from the definition of $\gamma^+$ where $D_1(\gamma^+) =  
D_2(\lambda - \gamma^+).$ The continuity at $\gamma^+$ is obvious
from the fact that $g_1(\gamma^+) = 0$ and $\lim_{\gamma_1 \rightarrow \gamma^+} ~g_1(\gamma_1) = 0.$
\end{proof}

\textbf{Lemma \ref{lemma:threshold}} \\
\begin{proof}
  Suppose $\Delta \geq g_1(0).$ From the definition of $\Delta$ and from
Eq. \eqref{eq:ggamma}, this implies that
\begin{eqnarray*}
c_1 - c_2 &\geq& b \left( D_2(\lambda) - D_1(0) \right) \\
&\geq& \beta \left( D_2(\lambda) - D_1(0) \right) 
\end{eqnarray*}
for all $\beta \in [a,b].$
From the Wardrop equilibrium condition, this implies that $K^W(\beta,\cdot) = \delta_2$ for  
$\beta \in [a,b].$ This implies that $\gamma_1 = 0$ and from
Eq. \eqref{eq:wardrop-cts} we have $\beta_1 = b.$
Similarly when, $\Delta \leq g_1(\lambda)<0$ we have
\begin{eqnarray*}
c_1 - c_2 &\leq& b \left( D_2(0) - D_1(\lambda) \right) \\ 
&\leq & \beta \left( D_2(0) - D_1(\lambda) \right)  
\end{eqnarray*}
where $\beta \in [a,b].$
Again, from the Wardrop equilibrium condition, this implies that
$K^W(\beta,\cdot) = \delta_1$ for  
$\beta \in [a,b].$ Hence $\gamma_1 = \lambda$  
and from Eq. \eqref{eq:wardrop-cts}, we have $\beta_1 = b.$

Now suppose $g_1(\lambda) < \Delta < g_1(0)$ where we know that $g_1(0) > 0$
and $g_1(\lambda) < 0.$
From Lemma \ref{lemma:gfx}, we know that $g_1(\gamma_1)$ is monotonically decreasing
in $\gamma_1.$ Therefore there exists a unique $\gamma$ with $0 < \gamma < \lambda$ 
such that $\Delta = g_1(\gamma).$ This proves the uniqueness of $\gamma_1.$
To see how $\beta_1 = \beta_1(\gamma)$ note that $\Delta = g_1(\gamma)$
implies that
\begin{equation}
 c_1 - c_2 = \beta_1(\gamma)\left(D_2(\lambda - \gamma) - D_1(\gamma) \right).  \nonumber
\end{equation}

Now if $\gamma \leq \gamma^+$ we have $D_2(\lambda - \gamma) > D_1(\gamma).$
In this case,
\begin{equation}
 c_1 - c_2 \leq \beta\left(D_2(\lambda - \gamma) - D_1(\gamma) \right)  \nonumber
\end{equation}
 for $\beta \in [a,\beta_1(\gamma)].$ This means that $K^W(\beta,\cdot) = \delta_2$
 for all $\beta \in [a,\beta_1(\gamma)].$ Similarly,  we have
 \begin{equation}
 c_1 - c_2 \geq \beta \left( D_2(\lambda - \gamma) - D_1(\gamma) \right)
\end{equation}
and $K^W(\beta,\cdot) = \delta_1$ when $\beta \in [\beta_1(\gamma),b].$
Similar arguments hold when $\gamma > \gamma^+$ and hence $\beta_1 = \beta_1(\gamma)$
when $g_1(0) < \Delta < g_1(\lambda).$

From Theorem \ref{thm:wardrop-cts}, $K^W$ is characterized by $\beta_1$ and for a fixed
$\Delta,$ $\beta_1$ is unique. This implies uniqueness of $K^W.$ 
It is important to mention that $K^W$ is unique when $\Delta = 0$ 
because of the assumptions made to ensure $\beta_1(\gamma_1)$ well defined
at $\gamma_1 = \gamma^+$.
\end{proof}

\textbf{Lemma \ref{lemma:gamma^1}}\\
\begin{proof}
Suppose $c_2$ satisfies $c_2 < - g_1(\lambda).$ Assume 
that $c_1 = 0$ so that we have $\Delta > g_1(\lambda).$ From Lemma \ref{lemma:threshold}
this implies that the equilibrium  $\gamma_1$ satisfies 
$g_1(\gamma_1) = \Delta = -c_2$. Let us label this $\gamma_1$ as $\hat{\gamma}.$
Now increase $c_1$ from $c_1 = 0$ by a small $\epsilon > 0$ such that  
there exists $\gamma_1$ that satisfies $\Delta = \epsilon - c_2 = g_1(\gamma_1).$
Now from the monotonicity of $g_1(\cdot)$ it is clear that the equilibrium $\gamma_1$ 
is decreasing as $\Delta$ increases. This implies that
a higher $\Delta$ caused by increasing $c_1$ will only lead to a $\gamma_1$
satisfying $\gamma_1 < \hat{\gamma}.$ Clearly, for any choice of $c_1 \geq 0,$
we have $\gamma_1 \notin [\hat{\gamma},\lambda]$ and hence for this case 
$\gamma^1(c_2) = \hat{\gamma}.$

Now suppose that $- c_2 \leq g_1(\lambda).$ When $c_1 = 0,$ this implies 
$\Delta \leq g_1(\lambda)$ and from Lemma \ref{lemma:threshold} this implies 
$\beta_1 = b$ with the corresponding $\gamma_1$ satisfying $\gamma_1 = \lambda.$
As we increase $c_1,$ the equilibrium $\gamma_1$ decreases and hence $\gamma_1$
satisfies $\gamma_1 \in [0, \lambda].$
The compact representation now follows.
\end{proof} 

\textbf{Lemma \ref{lem:gamma_1^*}} \\
\begin{proof}
To reduce the notations, we represent $\gamma_j^*(c_{j^-})$ by
$\gamma_j^*$ in the proof of the lemma. We shall prove that 
$\gamma_1^* \notin \left\lbrace 0, \gamma^1(c_2) \right\rbrace$
and the proof for 
$\gamma_2^* \notin \left\lbrace 0, \gamma^2(c_1) \right\rbrace$
is along similar lines.
Suppose $\gamma_1^* \in \left\lbrace 0, \lambda \right\rbrace.$
 Then from the requirement 
that $\gamma_1^* = \lambda - \gamma_2^*,$ we have either (1) $\gamma_1^*
= 0$ and $\gamma_2^* = \lambda$ or (2) $\gamma_1^* = \lambda$ and 
$\gamma_2^* = 0.$ First consider the case when  $\gamma_1^*
= 0$ and $\gamma_2^* = \lambda.$ This implies that $R_1(c_1(0),0) = 0$
and hence the revenue made by Server~1 at equilibrium is zero. 
Further since this is an equilibrium, there is no incentive for 
the server to change the admission price and increase its revenue.
We shall now show that this is not true. From Theorem 
\ref{thm:cgamma}, we know that for a given $c_2,$ the admission 
price at Server~$1$ must be at least $c_2 + g_1(0)  > 0.$ Now 
we know that setting $c_1 = c_2$ will result in $\gamma_1 = \gamma^+.$
Now due to the assumption that (1) $D_1(0) < D_2(\lambda)$
and (2)  $D_2(0) < D_1(\lambda),$ there exists an $\epsilon > 0$
such that  setting $c_1 = c_2 + \epsilon$ will result in 
$\gamma_1 \in (0, \gamma^+).$ The revenue earned is non-zero and 
there is clearly an incentive to deviate from any value greater than
$c_2 + g_1(0).$  This implies that $\gamma_1^* = 0$ and $\gamma_2^* = 
\lambda$ is not possible. The proof for $\gamma_1^*
= \lambda$ and $\gamma_2^* = 0$ is along the same lines.
\end{proof}

\bibliographystyle{elsarticle-num}
\bibliography{bodas}

\end{document}